\documentclass[11pt]{article}
\usepackage{geometry}
\usepackage{tikz}
\usepackage{tikz-cd}
\usetikzlibrary{matrix,arrows,shapes,cd}
\usepackage{pgf}
\usepackage{url}
\usepackage{longtable}
\usepackage{enumerate}
\usepackage{amsthm}
\usepackage{amssymb}
\usepackage{amsmath}
\usepackage{amscd}
\usepackage{eucal}
\usepackage{mathrsfs}
\usepackage{upgreek}
\usepackage{mathtools}
\usepackage{dsfont}
\usepackage{yfonts}
\usepackage[T1]{fontenc}
\usepackage{bbm}
\usepackage{graphics}
\usepackage{array}
\usepackage[retainorgcmds]{IEEEtrantools}
\usepackage{booktabs}
\usepackage{enumerate}
\usepackage{feynmf}
\usetikzlibrary{trees}
\usetikzlibrary{decorations.pathmorphing}
\usetikzlibrary{decorations.markings}
\tikzset{
	photon/.style={decorate, decoration={snake}, draw=red},
	electron/.style={draw=blue, postaction={decorate},
		decoration={markings,mark=at position .55 with {\arrow[draw=blue]{>}}}},
	gluon/.style={decorate, draw=magenta,
		decoration={coil,amplitude=4pt, segment length=5pt}},
	sderiv/.style={postaction={decorate},
		decoration={markings,mark=at position .3 with {\arrow{>}}}},
	tderiv/.style={postaction={decorate},
		decoration={markings,mark=at position .7 with {\arrow{<}}}},
	stderiv/.style={postaction={decorate},
		decoration={markings,mark=at position .7 with {\arrow{<}},mark=at position .3 with {\arrow{>}}}}
}
\definecolor{see}{RGB}{67,75,179}
\definecolor{darksee}{RGB}{42,44,148}
\definecolor{honey}{RGB}{232,180,129}
\definecolor{lighthoney}{RGB}{255,254,220}
\definecolor{citecol}{rgb}{0.5,0,0} 
\definecolor{blue1}{RGB}{130,150,209}
\DeclareSymbolFont{bbold}{U}{bbold}{m}{n}
\DeclareSymbolFontAlphabet{\mathbbold}{bbold}
\usepackage{soul}
\usepackage[colorinlistoftodos,textsize=footnotesize]{todonotes}

\setstcolor{red}
\definecolor{see}{RGB}{67,75,179}
\usepackage{hyperref}
\hypersetup
{
	bookmarks=true,
	bookmarksopen=true,
	pdftitle=PAQFT,
	pdfauthor=Kasia Rejzner,
	pdfsubject=An introduction for mathematicians, %subject of the document
	pdfmenubar=true, %menubar shown
	pdfhighlight=/O, %effect of clicking on a link
	colorlinks=true, %couleurs sur les liens hypertextes
	pdfpagemode=None, %aucun mode de page
	pdfpagelayout=SinglePage, %ouverture en simple page
	pdffitwindow=true, %pages ouvertes entierement dans toute la fenetre
	linkcolor=see, %couleur des liens hypertextes internes
	citecolor=red!75!darkgray, %couleur des liens pour les citations
	urlcolor=see %couleur des liens pour les url
}
\newcommand{\fA}{\mathfrak{A}}

\newcommand{\Gcal}{\mathcal{G}}  % gauge group
 
\newcommand{\Lcal}{\mathcal {L}}

\newcommand{\Bcal}{\mathcal {B}}

\newcommand{\Dcal}{\mathcal{D}}
\newcommand{\Ecal}{\mathcal{E}} 

\newcommand{\Fcal}{\mathcal{F}}

\newcommand{\Ocal}{\mathcal{O}}
\newcommand{\Scal}{\mathcal{S}}
\newcommand{\Pcal}{\mathcal{P}}
\newcommand{\Rcal}{\mathcal{R}}
\newcommand{\Tcal}{\mathcal{T}}

\newcommand{\Ci}{\mathcal{C}^\infty} % smooth functions
\newcommand{\WF}{\mathrm{WF}}         % wave front set
\newcommand{\id}{\mathrm{id}}               % identity
\newcommand{\loc}{\mathrm{loc}}
\newcommand{\Floc}{\Fcal_{\loc}}

\newcommand{\mc}{{\mu\mathrm{c}}}

\newcommand{\Diag}{\mathrm{Diag}}
\newcommand{\DIAG}{\mathrm{DIAG}}

\newcommand{\NN}{\mathbb{N}}          % natural naumbers
\newcommand{\RR}{\mathbb{R}}           % real  numbers
\newcommand{\CC}{\mathbb{C}}           % complex numbers

\newcommand{\MM}{\mathbb{M}} 	     % Minkowski spacetime
\newcommand{\al}{\alpha}

\newcommand{\Ga}{\Gamma}

\newcommand{\la}{\lambda}

\newcommand{\si}{\sigma}
\newcommand{\ph}{\varphi}

\newcommand{\TT}{\Tcal}

\newcommand{\DF}{D_{\mathrm{F}}}

\DeclareMathOperator{\Res}{\mathrm{Res}}
\newcommand{\1}{\mathds{1}}
\newcommand{\HFmu}{H^{\mathrm{F}}_{\sst M}}

\newcommand{\sst}[1]{\scriptscriptstyle{#1}}  % small font for the subscripts
\newcommand{\be}{\begin{equation}}
\newcommand{\ee}{\end{equation}}

\DeclareMathOperator{\supp}{\mathrm{supp}}      % support

\newcommand{\bGa}{\boldsymbol{\Gamma}}
 \author{Kasia Rejzner}
  \title{Renormalization and periods in perturbative Algebraic Quantum Field Theory}

\date{}

\begin{document}
\sloppy
 \maketitle
  \theoremstyle{plain}
  \newtheorem{df}{Definition}[section]
  \newtheorem{thm}[df]{Theorem}
  \newtheorem{prop}[df]{Proposition}
  \newtheorem{cor}[df]{Corollary}
  \newtheorem{lemma}[df]{Lemma}  
  \theoremstyle{plain}
  \newtheorem*{Main}{Main Theorem}
  \newtheorem*{MainT}{Main Technical Theorem}

  \theoremstyle{definition}
  \newtheorem{rem}[df]{Remark}
  \newtheorem{example}[df]{Example}

 \theoremstyle{definition}
  \newtheorem{ass}{\underline{\textit{Assumption}}}[section]
  %*****************
%\newpage
%*************************************************************
%\tableofcontents
%\markboth{Contents}{Contents}
%\newpage
%************************************************************
\begin{abstract}
	In this paper I give an overview of mathematical structures appearing in perturbative algebraic quantum field theory (pAQFT)  in the case of the massless scalar field on Minkowski spacetime.  I also show how these relate to Kontsevich-Zagier periods.
	 Next, I review the pAQFT version of the renormalization group flow and reformulate it in terms of Feynman graphs. This allows me to relate Kontsevich-Zagier periods to numbers appearing in computing the pAQFT $\beta$-function.
\end{abstract}
\tableofcontents
\section{Introduction}
Perturbative AQFT is a mathematically rigorous framework that allows to build models of physically relevant quantum field theories on a large class of Lorentzian manifolds. The basic objects in this framework are functionals on the space of field configurations and renormalization method used is the Epstein-Glaser (EG) renormalization \cite{EG}. The main idea in the EG approach is to reformulate the renormalization problem, using functional analytic tools, as a problem of extending almost homogeneously scaling distributions that are well defined outside some partial diagonals in $\RR^n$. Such an extension is not unique, but it gives rise to a unique ``residue'', understood as an obstruction for the extended distribution to scale almost homogeneously. Physically, such scaling violations are interpreted as contributions to the $\beta$ function.

The main result of this paper is Proposition~\ref{mainprop}, where we show how a large class of residues relevant for computing the $\beta$ function in the pAQFT framework, is related to Kontsevich-Zagier periods. Following \cite{KZ} we define:
\begin{df}
	A period is a complex number whose
	real and imaginary parts are values of absolutely convergent integrals of rational functions
	with rational coefficients, over domains in $\RR^n$ given by polynomial inequalities with
	rational coefficients.
\end{df}
A very accessible introduction to periods and their relation to Feynman integrals can be found for example in \cite{BognerThesis,BogWei07}.

In section 5 we review the main ideas behind the pAQFT renormalization group (following \cite{BDF}) and propose a reformulation in terms of Feynman graphs. The latter allows then to relate the numbers appearing in the computation of the pAQFT $\beta$ function to periods discussed in section 4.
\section{Functionals}
Let $\MM$ be the $D$-dimensional Minkowski spacetime, i.e. $\RR^D$ with the metric \[\eta=\operatorname{diag}(1,\underbrace{-1,\dots,-1}_{D-1})\,.
\] Define the \emph{configuration space} $\Ecal$ of the theory as the space of smooth sections of a vector bundle $E$ over $\MM$, i.e. $\Ecal\doteq \Gamma(E\xrightarrow{\pi}\MM)$. Fixing $E$ specifies the particle content of the model under consideration. In this paper we will consider only the scalar field, i.e. $\Ecal=\Ci(\MM,\RR)$. The field configurations are denoted by $\ph$. For future reference, define $\Dcal\doteq\Ci_c(\MM,\RR)$ the space of smooth compactly supported functions on $\MM$ and more generally, $\Dcal(\Ocal)\doteq\Ci_c(\Ocal,\RR)$, where $\Ocal$ is an open subset of $\RR^n$.

Let $\Ci(\Ecal,\CC)$ denote the space of smooth \cite{Bas64,Neeb} functionals on $\Ecal$. An important class of functionals is provided by the local ones.
\begin{df}
	A functional $F\in\Ci(\Ecal,\CC)$  is called local (an element of $\Fcal_{\loc}$) if for each $\ph\in\Ecal$ there exists $k\in\NN$ such that
	\begin{equation}
	F(\ph)=\int_{\MM} f(j^k_x(\ph))\ ,
	\end{equation}
	where $j^k_x(\ph)$ is the $k$-th jet prolongation of $\ph$ and $f$ is a density-valued function on the jet bundle.
\end{df}
The following definition introduces the notion of spacetime localization of a functional.
\begin{df}
The spacetime support $\supp F$ of a functional $F\in\Ci(\Ecal,\CC)$ is defined by
\begin{multline*}
\label{support}
\supp F\doteq\{ x\in \MM|\forall \text{ neighborhoods }U\text{ of }x\ \exists \ph,\psi\in\Ecal, \supp\,\psi\subset U\,,
\\ \text{ such that }F(\ph+\psi)\not= F(\ph)\}\ .
\end{multline*}	
\end{df}
Derivatives of smooth compactly-supported functionals are distributions with compact support\footnote{Prime always denotes the topological dual, so $\Ecal'(\MM^n)$ is the space of continuous linear maps from $\Ecal(\MM^n)$ to $\RR$ and similarly, $\Ecal'(\MM^n,\CC)$ is the space of continuous linear maps to $\CC$. $\Ecal(\MM^n)$ is always understood as equipped with its natural Fr\'echet topology. It is a standard result in functional analysis that the dual of the space of smooth functions is exactly the space of distributions with compact support.}, i.e.
\[
F^{(n)}(\ph)\in\Ecal'(\MM^n,\CC)\equiv{\Ecal'}^{\sst{\CC}}(\MM^n)\,,\quad\forall\ph\in\Ecal\,,n\in\NN\,.
\]
	If $F$ is local then each $F^{(n)}(\ph)$ is a distribution supported on  the thin diagonal
	\be\label{thindiagonal}
	D_n\doteq\{(x_1,\dots,x_n)\in \MM^n, x_1=\dots=x_n\}\,.
	\ee
Local functionals are important, since they are used to model interactions in perturbative QFT. In the Epstein-Glaser approach, interaction is first restricted to a compact region to avoid the IR problem and subsequently extended by taking the \textit{adiabatic limit}. In this work we are interested only in the UV (i.e. short distance) behavior of the theory, so we leave this last step out. 

One can define various important classes of functionals by formulating conditions on the singularity structure of their derivatives $F^{(n)}(\ph)\in{\Ecal'}^{\sst{\CC}}(\MM^n)$. A notion used in this context is that of a \textit{wavefront set}. For a given distribution $u\in\Dcal'(\RR^n)$, its wavefront set $\WF(u)$ contains information about points in $\RR^n$ at which $u$ is singular, but also about directions in the momentum space (i.e. after the Fourier transform) in which $\widehat{u}(k)$ fails to decay sufficiently fast. In other words, $\WF(u)$ characterizes \textit{singular directions} of $u$. For a pedagogical introduction to WF sets see \cite{BrDang14}. Knowing the WF sets of distributions $u_1$, $u_2$  one can apply the criterion due to H{\"o}rmander \cite{Hoer1} to check if the pointwise product of  $u_1$, $u_2$ is well defined. This motivates using WF sets of functional derivatives $F^{(n)}(\ph)$ to distinguish classes of ``well-behaving'' functionals. One such class is called \textit{microcausal functionals} $\Fcal_\mc$. For the precise definition see \cite{BDF} and \cite{Book} for possible modifications of this notion. For the purpose of this paper, it is enough to know that $\Fcal_{\loc}\subset \Fcal_{\mc}$ and that some important algebraic structures are well defined on this space.
%\begin{df}[after \cite{BDF}]\label{GenLagr}
%	A generalized Lagrangian is a smooth (in the sense of Bastiani \cite{Bas64,Ham,Mil}) map $\hat{\Psi}:\Dcal \rightarrow \Floc[[\hbar]]$ satisfying:
%	\begin{enumerate}
%		\item {\bf Additivity}: $\hat{\Psi}(f_1+f_2+f_3)=\hat{\Psi}(f_1+f_2)+\hat{\Psi}(f_2+f_3)-\hat{\Psi}(f_2)$, if $\supp f_1\cap\supp f_3=\varnothing$.
%		\item {\bf Support}: $\supp\hat{\Psi}(f)\subset\supp f$.
%		\item {\bf Covariance}: $\hat{\Psi}(f)(L^*\ph)=\hat{\Psi}({L}_*f)(\ph)$ for every element $L$ of the proper ortochronous Poincar\'e group $\Pcal^\uparrow_+$.
%		\item {\bf Zero point}: $\hat{\Psi}(0)=0$.
%	\end{enumerate}
%\end{df}
%Actions are equivalence classes of Lagrangians.
%\begin{df}
%	\begin{equation}\label{equiv}
%	L_1\sim L_2\,,\quad \mathrm{iff}\ \supp(L_1-L_2)(f)\subset (f^{-1}(1))^c\,,
%	\end{equation}
%	where $L_1,L_2\in\mathscr{L}$.
%\end{df}
%\todo{Example of a Lagrangian, something on renormalizability}
\section{The S-matrix and time-ordered products}
In the next step we introduce the S-matrix. Since we work perturbatively, the S-matrix is understood as a formal power series in the coupling constant $\la$ and a Laurent series in $\hbar$, with coefficients in smooth functionals. First we introduce the time-ordered products.
\begin{df}\label{Tns1}
	Time ordered products are multilinear maps $\TT^n:\Floc^{\otimes n}\rightarrow \Fcal_{\mc}[[\hbar]]$, $n\in\NN$, satisfying:
	\begin{enumerate}
		\item{\bf Causal factorisation property}
		\[\TT^n(F_1,\dots,F_n)=\TT^k(F_1,\dots,F_k)\star\TT^{n-k}(F_{k+1},\dots,F_n)\,,\]
		if the supports $\supp F_i$, $i=1,\dots,k$ of the first $k$ entries do not intersect the past of the supports $\supp F_j$, $j=k+1,\dots,n$ of the last $n-k$ entries. Here $\star$ is the operator product of the quantum theory defined by
		\[
		(F\star G)(\ph)\doteq e^{\hbar\left\langle\Delta_+,\frac{\delta^2}{\delta\ph\delta\ph'}\right\rangle}F(\ph)G(\ph')|_{\ph'=\ph}\,,
		\]
		where $\Delta_+$ is the Wightman 2-point function.
		\item $\TT^0=1$, $\TT^1=\id$.
		\item {\bf Symmetry}: For a purely bosonic theory $\TT^n$s are symmetric in their arguments. If the fermions are present,  $\TT^n$s are graded-symmetric.
		
		\item {\bf Field independence}: $\TT^{n}(F_1,\ldots,F_n)$, as a functional on $\Ecal$, depends on $\ph$ only via the functional derivatives of $F_1,\ldots,F_n$, i.e.\label{FieldIndep}
		\[
		\frac{\delta}{\delta\ph}\TT^{n}(F_1,\ldots,F_n)=\sum_{i=1}^n\TT^{n}\left(F_1,\dots,\frac{\delta F_i}{\delta\ph},\dots,F_n\right)
		\]
		\item {\bf $\ph$-Locality}: $\TT^{n}(F_1,\ldots,F_n)=\TT^{n}(F_1^{[N]},\ldots,F_n^{[N]})+\Ocal(\hbar^N)$, where  $F_i^{[N]}$ is the Taylor series expansion of the functional $F_i$ up to the $N$-th order.\label{PhLoc}
		\item {\bf Poincar\'e invariance}. Let $\alpha\in\Pcal^\uparrow_+$ (the proper ortochronous Poincar\'e group). We define $\sigma_\al(\ph)(x)\doteq \ph(\al^{-1}x)$ for $\ph\in\Ecal$, $x\in\MM$ and define the action of $\alpha\in\Pcal^\uparrow_+$ on functionals using $\sigma_\al(F)\doteq F( \sigma_\al(\ph))$. We require $\sigma_\al\circ \Tcal^{n}\circ (\sigma_\al^{-1})^{\otimes n}=\Tcal^n$.
	\end{enumerate}
\end{df}
We refer to these conditions as the Epstein-Glaser (EG) axioms.
\begin{df}
The formal S-matrix is a map from $\Fcal_{\loc}$ to $\Fcal_{\mc}[[\la]]((\hbar))$ defined as
 \begin{equation}\label{S}
\Scal(\la F)=\sum_{n=0}^\infty \frac{(\la i)^n}{n!\hbar^n}\TT_n(F^{\otimes n})\, ,
\end{equation}
\end{df}
Let $(\Fcal_\loc)^{\otimes n}_{\mathrm{pds}}$ denote the subset of $\Fcal_\loc^{\otimes n}$ consisting of functionals with pairwise disjoint supports. On such functionals one can define the $n$-fold time-ordered product to be
\be\label{nonrenorm}
\TT^n(F_1,\dots,F_n)=m\circ e^{\hbar\sum_{i<j}D_{\mathrm{F}}^{ij}}(F_1\otimes\dots\otimes F_n)\,,
\ee
where $D_{\mathrm{F}}^{ij}\doteq \langle\Delta^{\mathrm{F}},\frac{\delta^2}{\delta\ph_i\delta\ph_j}\rangle$, $m$ denotes the pointwise multiplication and $\Delta^{\mathrm{F}}$ is the Feynman propagator of the free scalar field theory on $\MM$. Unfortunately, this definition doesn't trivially extend to arbitrary local functionals, due to singularities of the Feynman propagator. Instead, one has to use more sophisticated analytical tools, which we will review in the next section. We will refer to \eqref{nonrenorm} as the \textit{non-renormalized} $n$-fold time-ordered product and the problem of extending $\Tcal^n$ to arbitrary local functional is referred to as \textit{the renormalization problem}.

To organize the combinatorics present in the construction of time-ordered products, it is convenient to write them in terms of Feynman graphs. To see how this comes about, we use the identity
\be\label{expDij}
e^{\hbar\sum_{i<j}D_{\mathrm{F}}^{ij}}=\prod_{i<j}\sum_{l_{ij}=0}^{\infty}\frac{\left(\hbar\, D_{\mathrm{F}}^{ij}\right)^{l_{ij}}}{l_{ij}!}
\ee
 %and $\Dcal_{\mathrm{F}}\doteq\langle\Delta_{S_0}^{\mathrm{F}},\frac{\delta^2}{\delta\ph^2}\rangle$.
to obtain the expansion 
\[
\Tcal^n=\sum_{\Gamma\in\Gcal_n}\Tcal^\Gamma\,,
\] where $\Gcal_n$ is the set of all graphs with $n$ vertices and no tadpoles (i.e. no loops in the graph-theoretic sense). Let $E(\Gamma)$ denote the set of edges and $V(\Gamma)$ the set of vertices of the graph $\Gamma$. Contributions from particular graphs are given by
\be\label{GraphDO}
\Tcal^{\Gamma}=\frac{1}{\textrm{Sym}(\Gamma)}m\circ\langle t^{\Gamma},\delta_{\Gamma}\rangle\,,
\ee
with
\[\delta_{\Gamma}=\frac{\delta^{2\,|E(\Gamma)|}}{\prod_{i\in V(\Gamma)}\prod_{e:i\in\partial e}\delta\ph_i(x_{e,i})}\]
and
\be\label{SGamma}
t^{\Gamma}=\prod_{e\in E(\Gamma)}\hbar\Delta^{\mathrm{F}}(x_{e,i},i\in\partial e)
\ee
The symmetry factor $\textrm{Sym}$ is the number of possible permutations of lines joining 
the same two vertices, $\textrm{Sym}(\Gamma)=\prod_{i<j}l_{ij}!$. 

Note that the map $\delta_\Gamma$ applied to $F\in\mathcal{F}_{\mathrm{loc}}^{\otimes n}$ yields, at any $n$-tuple of field configurations $(\varphi_1,\dots,\varphi_n)$, a compactly supported distribution in the variables $x_{e,i},i\in\partial e, e\in E(\Gamma)$ with support on the partial diagonal 
\[
\Diag_{\Gamma}=\{x_{e,i}=x_{f,i},i\in\partial e\cap\partial f, e,f\in E(\Gamma)\}\subset \MM^{2|E(\Ga)|}\,.
\]
This partial diagonal can be parametrized using the centre of mass coordinates
\[
z_v\doteq \frac{1}{\textrm{valence}(v)}\sum_{e:v\in\partial e} x_{e,v}\,,
\]
assigned to each vertex. The remaining relative coordinates are $x_{e,v}^{\text{rel}}=x_{e,v}-z_v$, where $v\in V(\Ga)$, $e\in E(\Ga)$ and $v\in\partial e$. Obviously, we have $\sum_{e|v\in\partial e} x_{e,v}^{\text{rel}}=0$ for all $v\in V(\Ga)$, so in fact $\Diag_{\Gamma}$ is parametrized by $|V(\Gamma)|-1$ independent variables. In this parametrization $\delta_\Gamma F$ can be written as a finite sum
\[
\delta_\Gamma F=\sum_{\beta}f^\beta\partial_\beta\delta_{\textrm{rel}}\,,
\]
where $\beta\in\NN_0^{D(|V(\Gamma)|-1)}$, each $f^{\beta}(\ph_1,...,\ph_n)$ is a test function on $\Diag_{\Gamma}$ and $\delta_{\textrm{rel}}$ is the Dirac delta distribution in relative coordinates, i.e. $\delta_{\textrm{rel}}(g)=g(0,\ldots,0)$, where $g$ is a function of $(x_{e,v}^{\textrm{rel}},v\in V(\Ga), e\in E(\Ga))$.

Let $Y_\Gamma$ denote the vector space spanned by derivatives of the Dirac delta distributions $\partial_\beta\delta_{\textrm{rel}}$, where $\beta\in\NN_0^{D(|V(\Gamma)|-1)}$ and let $\mathcal{D}(\Diag_\Gamma, Y_\Gamma)$ denote the graded space of test functions on $\Diag_\Gamma$ with values in $Y_\Gamma$. With this notation we have $\delta_\Gamma F\in\mathcal{D}(\Diag_\Gamma, Y_\Gamma)$ and if $F\in (\Fcal_\loc)^{\otimes n}_{\mathrm{pds}}$, then $\delta_\Gamma F$ is supported on  $\Diag_{\Gamma}\setminus\DIAG$, where $\DIAG$ is the large diagonal:
\[
\DIAG=\left\{ z\in\Diag_{\Gamma}|\,\exists v,w\in V(\Gamma),v\neq w:\, z_{v}=z_{w}\right\} \,.
\]
We can therefore write \eqref{GraphDO}  in the form
\[
\frac{1}{\textrm{Sym}(\Gamma)}\langle t^{\Gamma},\delta_{\Gamma}\rangle=\sum_{\textrm{finite}}\left<f^\beta\partial_\beta\delta_{\textrm{rel}},t^\Gamma\right>
\]
where $t^{\Gamma}$ is written in terms of centre of mass and relative coordinates. To see that this expression is well defined, note that we can move all the partial derivatives $\partial_\beta$ to $t^{\Ga}$ by formal partial integration. Then the contraction with $\delta_{\textrm{rel}}$ is just the pullback through the diagonal map $\rho_{\Ga}:\Diag_\Gamma\rightarrow\MM^{2|E(\Ga)|}$ by
\[
(\rho_{\Ga}(z))_{e,v}=z_v\,\quad\mathrm{if}\,v\in\partial e\,.
\]
The pullback $\rho_{\Ga}^*$ of each $ t^{\Ga}_\beta\doteq\partial_\beta t^{\Ga}$ is a well defined distribution on
$\Diag_\Ga\backslash\DIAG$, so  \eqref{GraphDO} makes sense if $F\in (\Fcal_\loc)^{\otimes n}_{\mathrm{pds}}$.

The renormalization problem to extend $\TT^n$'s to maps on the full $\Fcal_\loc^{\otimes n}$ is now reduced to extending distributions $\rho_{\Ga}^*t_{\beta}^\Gamma$ to the diagonal.

In this and the next section we will consider the simplest situation, where the free theory is the free massless scalar field and the possible interactions are local functionals $F_1,\dots,F_n$ that depend on the field itself but not on its derivatives. Without the loss of generality, we can assume them to be monomials, i.e. of the form
\[
F(\ph)=\int f(x)\ph(x)^ld^Dx\,,
\]
where $f\in\Dcal$, $l\in\NN$. Such a functional can be graphically represented as a vertex of valence $l$, decorated by the test function $f$. 

The distributions we need to extend are then $u^\Gamma=\rho_{\Ga}^*t^{\Gamma}$, where $t^{\Gamma}$ is given by \eqref{SGamma}. We can write the explicit expression for $u^\Gamma$ using the following rules:
\begin{enumerate}
	\item Choose a vertex of $\Gamma$ and label it as $x_0=0$. Label the remaining vertices with variables $x_1,\dots,x_{n}$, where $n=|V(\Gamma)|-1$.
	\item Assign the Feynman propagator $\Delta^{\mathrm{F}}(x_i,x_j)$ to each edge $e\in E(\Gamma)$, where $x_i,x_j\in\partial e$.
\end{enumerate}
Because of the translational symmetry, the Feynman propagator $\Delta^{\mathrm{F}}(x,y)$  depends only on the difference $x-y$. Explicitly, it is given by
\[
\Delta^{\mathrm{F}}(x,y)=(-1)^{\frac{D}{2}-1}\frac{\bGa(\frac{D}{2}-1)}{4\pi^{\frac{D}{2}}}\lim_{\epsilon\rightarrow 0^+}\frac{1}{((x-y)^2-i\epsilon)^{\frac{D}{2}-1}}\equiv \frac{k_D}{((x-y)^2-i0)^{\frac{D}{2}-1}}\,,
\] 
where $(x-y)^2\doteq \eta(x-y,x-y)$ is the square with respect to the Minkowski metric and $\bGa$ denotes the Gamma function. We use the bold symbol to distinguish this from the notation we use for graphs. It follows now that
\be\label{uGamma}
u^\Gamma(x_1,\dots,x_{n-1})=\frac{k_D^{|E(\Gamma)|}}{\prod_{e\in E(\Gamma)}((x_{s(e)}-x_{f(e)})^2-i0)^{\frac{D}{2}-1}}\,,
\ee
where $\{x_{s(e)},x_{f(e)}\}=\partial e$ is the pair of vertices that constitute the boundary of an edge $e$ and the order of these vertices is irrelevant. 
\begin{example}
	Consider the following examples:
	\begin{enumerate}
		\item For the fish graph: $u^\Gamma(x)= \frac{k_D^2}{(x^2-i0)^{D-2}}$,
		\item For the triangle graph: \[u^\Gamma(x,y)=\frac{k_D^3}{(x^2-i0)^{\frac{D}{2}-1}(y^2-i0)^{\frac{D}{2}-1}((x-y)^2-i0)^{\frac{D}{2}-1}}\,.\]
	\end{enumerate}
\end{example}
We have seen how to reduce the renormalization problem to extension of distributions. The construction of $\Tcal^n$s proceeds inductively. Given renormalized time-ordered products of order $k<n$, we can use the causal factorisation property  to fix the time-ordered products at order $n$ up to the thin diagonal $D_n$ (see \eqref{thindiagonal}). On the level of graphs it means that all the distributions $u^\gamma$ corresponding to proper subgraphs $\gamma\subset\Gamma$ have been constructed and substituted into $u^\Gamma$. The renormalization problem for $u^\Gamma$ is now the extension of a distribution defined everywhere outside the \textit{thin diagonal of the graph $\Gamma$} understood as the subset of $\Diag_\Gamma$ with all the variables equal. Because of the translation symmetry, this is in fact extension problem for a distribution defined everywhere outside the origin.
\section{Distributional residues and periods}\label{resid}
The framework of pAQFT is different from the one of Connes and Kreimer in two fundamental ways: one works in position rather than momentum space and the metric of the underlying spacetime has Lorentzian rather than Euclidean signature. The latter is the reason for invoking Epstein-Glaser causal approach to renormalization, as outlined in the previous section. 

There has been a lot of work done concerning periods in position space approach to renormalization. The most recent comprehensive review has been given in \cite{NST}, while for historical remarks on the development of the subject, it is worth to look up \cite{Tod16}. A very detailed analysis of renormalization of Feynman integrals and its relation to periods and motives has been done in the series of papers \cite{CeMar12,CeMar12b,CeMar13}. However, the computations performed in these works are done in Euclidean signature. Another noteworthy work, focusing on relations between Epstein-Glaser renormalization and ``wonderful compactifications'' is \cite{BBK09}.

There are some serious technical difficulties arising when changing the signature to Lorentzian. In the present paper we show how some standard methods used in Euclidean setting can, nevertheless, be applied also to the Lorentzian case.

Before coming to the main result of this paper, let us recall some basic facts about the problem of extension of almost homogeneous distributions \cite{Ste71,BF0,HW02,JMGB03,BDF,NST}.
\begin{df}
We say that a distribution $u\in\Dcal'(\RR^N\setminus\{0\})$ scales almost homogeneously, if $(\rho\frac{d}{d\rho})^{k+1}\rho^\alpha u(\rho .)=0$ for some $k\in\NN_0$ (called scaling order), $\alpha\in \RR$ (called scaling degree).
\end{df}
The almost homogeneous scaling relation can also be written in terms of the Euler operator $\mathscr{E}=\sum_{i=1}^d x^i\frac{\partial}{\partial x^i}$, namely a distribution with scaling degree $\al$ and order $k$ satisfies
\[
(\mathscr{E}+\al)^{k+1}u=0\,,
\]
while $(\mathscr{E}+\al)^{k}u\neq0$.
\begin{example}
	For a graph $\Gamma$ with $n$ vertices the distribution $u\equiv u^\Gamma$ that we need to extend belongs to $\Dcal'(\RR^{N}\setminus\{0\})$, where $N=(n-1)D$ and $D$ is the dimension of $\MM$.
\end{example}
The following result  was proven in \cite[Proposition 1]{H} (see also \cite[section 4.4]{NST}):
\begin{prop}
 Let $u$ be a (Lorentz invariant) almost homogeneously scaling distribution with degree $\al=N+\NN_0$, then there  exists a non-unique (Lorentz invariant) extension $\bar{u}\in\Dcal'(\RR^N)$ of $u$ and 
 \[
 \left(\rho\frac{d}{d\rho}\right)^{k+1}\rho^\al \bar{u}(\rho .)\Big|_{\rho=1}=(\mathscr{E}+\al)^{k+1}\bar{u}=\sum_{|\beta|=\al-N}c_\beta\partial^\beta\delta\,,
 \]
 where $\beta\in\NN_0^N$ is a multiindex. 
 \end{prop}
In the proof of the above proposition provided in \cite{H}, the coefficients $c^{\alpha}$ are computed by integrating certain (closed) distributional forms over a closed codimension 1 surface enclosing the origin. We will now review the construction of these forms and it will become clear that these do not depend on the choice of the extension. Moreover, their closeness is the reason why $c^{\alpha}$s do not depend on the choice of the integration surface and hence the homogeneous differential operator
 \be\label{homdiffop}
\sum_{|\beta|=\al-N}c_\beta\partial^\beta
\ee
doesn't depend on the choice of the extension $\bar{u}$. This fact has also been highlighted in the discussion following formula (4.21) in \cite[section 4.4]{NST}).

 We will call \eqref{homdiffop} the residue of $u$ and denote it by $\operatorname{Res}(u)$, so that
 \[
 (\mathscr{E}+\al)^{k+1}\bar{u}=\operatorname{Res}(u)\delta\,.
 \]
Coefficients of the differential operator $\operatorname{Res}(u)$ can be explicitly computed using the construction of $\bar{u}$ proposed in \cite[eq.~(186)]{H} and \cite[Theorem 4.8]{NST}. Let us outline the main ideas behind this construction. First, note that the almost homogeneous scaling implies that the distributional kernel of $u$ can be written as \cite[eq.~(172)]{H}, \cite[eq.~(3.12)]{NST})
\be\label{expansion}
u(rx)=\sum_{m=0}^{k}
  r^{-l}\frac{(\log r)^m}{m!}v_m(x)\,\quad r>0\,,
\ee
where $v_m=(\mathscr{E}+\al)^mu$.
Let $\left<u,f\right>$ denote the dual pairing between the distribution $u$ and the test function $f\in\Dcal(\RR^{N}\setminus\{0\})$. This pairing is usually realized as the integral
\be\label{pairing}
\left<u,f\right>=\int_{\MM^N}u(x)f(x)d^{N}x\,.
\ee
We rewrite this integral using the representation \eqref{expansion}. First, choose a compact $N-1$ dimensional hypersurface around the origin, homoeomorphic to the (Euclidean) sphere $S^{N-1}$ that intersects each orbit of the scaling transformation $x\mapsto \mu x$ exactly once. Note that the map $\RR_+\times\Sigma\ni(r,\hat{x})\mapsto r\hat{x}\in\RR^N\setminus\{0\}$ is a diffeomorphism, since the surface $\Sigma$ is transverse to
the orbits of dilations in $\RR^{N}$. 

Using microlocal analysis techniques \cite{Hoer1} one can show that distributions $v_m$ appearing in \eqref{expansion} have well defined restrictions to $\Sigma$ (see \cite{H}, section 3.3, after eq.~(173)). Denote points on $\Sigma$ by $\hat{x}$ and write the restriction of $v_m$ as $v_m(\hat{x})$. Next, define for $r>0$ the following space
\[
\Sigma_r\doteq\{r\hat{x}\in\RR^{N}|\hat{x}\in\Sigma\}\,.
\]
Denote the natural inclusion of $\Sigma_r$ into $\RR^N$ by $i_r$. One introduces a $(N-1)$-form $\Omega$ on $\RR^N$ by
\[
\Omega(x)=\sum_{a=1}^{N}(-1)^{a-1}x_{a} dx_{1}\wedge\dots\wedge \widehat{dx_{a}}\wedge\dots\wedge dx_{N}\,,
\] 
where $x_{a}$ are components of $x\in\RR^N$. The caret symbol $\widehat{\,}$ means that the corresponding factor is omitted. We can now write
\[
d^{N}x=\frac{dr}{r}\wedge i_r^*\Omega\,.
\]
Let $\rho_\Sigma:\RR^{N}\setminus\{0\}\rightarrow \RR_+$ denote the smooth function defined by the condition
\[
\frac{x}{\rho_{\Sigma}(x)}\in\Sigma\,.
\]
We obtain a measure on $\Sigma$ by setting
\[
d\sigma(\hat{x})=\rho_{\Sigma}(x)^{-N}\Omega(x)\,,
\]
and express the pairing \eqref{pairing} as
\be\label{pairing2}
\left<u,f\right>=\int_{-\infty}^\infty \sum_{m=0}^{k}\theta(r)r^{N-1-l}\frac{(\log r)^m}{m!}\left(\int_{\Sigma}v_m(\hat{x})f(r\hat{x})d\sigma(\hat{x})\right)dr\,,
\ee
where $\theta$ denotes the Heaviside step function. 
Denote $F(r)\doteq \int_{\Sigma}v_m(\hat{x})f(r\hat{x})d\sigma(\hat{x})$.
Formula \eqref{pairing2} makes sense, since the support of the test function $f$ is bounded away from the origin
in $\RR^N$ and hence $F(r)$
is a test function on $\RR_+$ (i.e. smooth compactly supported), whose support is
 bounded away from $r=0$. If we want $f$ to be an arbitrary test function, then $F(r)$ vanishes for
 sufficiently large $r$, but does not vanish near $r = 0$ \cite[discussion following eq.~(184)]{H}. 
 
 The renormalization problem has therefore been reduced to extension of the distribution $\theta(r)r^{N-1-l}(\log r)^m$ on $\RR$. This is done by various methods, see for example \cite{Sp,FJL92,GBL,NST,GGV14}. The idea that we are going to follow here (proposed by \cite{JMGB03} based on the ideas of \cite{OK89,Pra99}) is to consider first the extension of the distribution $\theta(r)r^{N-1-l+\varepsilon}(\log r)^m$ for a complex, non integer $N-1-l+\varepsilon$. If we require the almost homogeneous scaling, then the extension exists and is unique. Next, we expand the resulting extended distribution in $\varepsilon$ and subtract the pole part.

 Let us come back to our original extension problem for $u\in\Dcal'(\RR^N\setminus\{0\})$. It is well known in the literature on differential renormalization
(see e.g. \cite[eq.~(186)]{H} or \cite[Thm.~4.8]{NST}) 
that an extension $\bar{u}$ of an almost homogeneously scaling distribution $u$ of order $k$ and degree $\al$ to an everywhere-defined distribution can be obtained by setting
\begin{multline*}
\left<\bar{u},f\right>\doteq\lim_{\varepsilon\rightarrow0} \left(\int_0^\infty\!\!\!\int_{\Sigma} \overline{r^{\varepsilon}u(r\hat{x})}^{\,\textrm{uhe}}f(r\hat{x})d\sigma(\hat{x})dr\right.\\-\left.\sum_{m=0}^{k}\frac{(-1)^{m+\al-N}}{\varepsilon^{m+1}}\sum_{|\beta|=\al-N}\frac{1}{\beta!}\int_0^\infty\!\!\!\int_{\Sigma}v_m(\hat{x})f(r\hat{x})\partial^{\beta}\delta(r\hat{x})d\sigma(\hat{x})dr\right)\,,
\end{multline*}
where $\overline{\phantom{u}\!\!\!.}^{\,\textrm{uhe}}$ denotes the \textit{unique almost homogeneous extension}, $\beta\in\NN_0^{N}$ is a multiindex, $\beta!\equiv \beta_1!\dots\beta_N!$ and $\partial^{\beta}\doteq \partial_{x_1}^{\beta_1}\dots\partial_{x_N}^{\beta_N}$.

We are now ready to compute the almost homogeneous scaling violation for the extension $\bar{u}$. The coefficients $c_\beta$ of $\operatorname{Res}(u)$ in formula \eqref{homdiffop} are obtained from (see e.g. \cite[eq.~(92)]{H})
\[
c_\beta\doteq (-1)^{\al-N}\frac{1}{\beta!}\int_{\Sigma}\hat{x}^{\beta}v_m(\hat{x})d\sigma(\hat{x})
\]
that manifestly doesn't depend on the choice of the extension, but only on $u$. Note that $c_\beta$ does not depend on the choice of $\Sigma$ because the integrand is a (distributional) closed form (see \cite[eq.~(210)]{H} for the proof of closedness).

As a special case we can consider a distribution with scaling degree $\al=N$ and scaling order 0. In this case the  residue is given in terms of a complex number
\be\label{c0}
\operatorname{Res}(u)=c_0=\int_{\Sigma}u(\hat{x})d\sigma(\hat{x})\,.
\ee
\begin{df}
For a graph $\Gamma$ with $n$ vertices and no derivatives decorating the edges, the scaling degree of the distribution $u^\Ga$ is given by the formula
\[
\al_{\Gamma}=(D-2)|E(\Gamma)|\,.
\]
\end{df}
\begin{df}
We define the divergence degree of a graph $\Gamma$ by
\[
\omega_{\Gamma}=\al_{\Gamma}-(|V(\Gamma)|-1)D\,.
\]
A graph $\Gamma$ is called \textit{superficially divergent} if $\omega_{\Gamma}\geq 0$.
\end{df}
Hence graphs with $\al_\Ga=N$ are characterized by the condition
\be\label{cond1}
(D-2)|E(\Gamma)|=(|V(\Gamma)|-1)D\,.
\ee
Note that the loop number of a graph (the first Betti number) is given by $h_1=|E(\Gamma)|-|V(\Gamma)|+1$, so the above condition can be also expressed as
\[
|E(\Gamma)|=\frac{D}{2}h_1\,.
\]
In four dimensions ($D=4$) this reduces to $|E(\Gamma)|=2h_1$.  If $\Gamma$ satisfies \eqref{cond1} and has no superficially  divergent subgraphs (here a subgraph $\gamma\subset\Gamma$ is specified by choosing a subset of vertices of $\Gamma$ and taking all the edges connecting these), then it has scaling degree $\al_{\Gamma}=N$ (so the divergence degree vanishes) and scaling order $k_{\Gamma}=0$. Such graphs coincide with \textit{primitive graphs} in the Connes-Kreimer approach, if we restrict to $D=4$ and fix the interaction.
\begin{rem}
	The class of primitive graphs in the Epstein-Glaser Hopf algebra \cite{Pint,Kai,GBL,DFKR} differs from the class of   primitive graphs in the Connes-Kreimer approach. As an example consider the two vertex graph, which has $|E(\Gamma)|=4$ and $h_1=3$. This graph is primitive in the Epstein-Glaser Hopf algebra, but not primitive in the Connes-Kreimer approach.
\end{rem}
Consider a graph $\Gamma$ with $|E(\Gamma)|=\frac{D}{2}h_1$ and no superficially divergent subgraphs. Let 
$\Delta$ be the simplex defined by $\sum_{e\in E(\Gamma)}\al_e=1$ and $\alpha_e> 0$. We introduce the measure $\mu(\vec{\al})\doteq\delta(1-\sum_{e\in E(\Gamma)}\al_e)\prod_{e\in E(\Gamma)}\al_e^{\frac{D}{2}-2}d\al_e$ on $\Delta$. Let
\[
\hat{\Psi}_\Gamma(\vec{\al})=\sum_{T\ \mathrm{spanning}\atop \mathrm{tree}}\prod_{e\in T}\al_e
\]
be the dual graph polynomial (see e.g. \cite{BEK06,BognerThesis,Wei06,IZ}). We define
\be\label{period}
P_\Gamma\doteq\int_{\Delta}\frac{\mu(\vec{\al})}{(\hat{\Psi}_\Gamma(\vec{\al}))^{D/2}}\,.
\ee
If $P_\Gamma$ converges absolutely, then it defines a \textit{real period} of the graph $\Gamma$ in the sense of Definition~36 of \cite{Brown}. 
%\begin{df}[Definition~36 of \cite{Brown}]
%	We define the real periods of the graph $\Gamma$, denoted $\fP(\Gamma)$, to be the $\Q$-vector
%	space spanned by the real numbers
%	\[
%	\int_{\Delta} \frac{P(\vec{\al})}{\Psi^n_\Gamma}\delta(H)\prod_{e\in E(\Gamma)}d\al_e\,,
%	\]
%	where $n\in\NN$, $P\in\Q[\vec{\al}]$ such that the integral converges absolutely.
%\end{df}

It was shown in \cite{BEK06} that, in $D=4$, under assumptions on $\Gamma$ stated above, $P_\Gamma$ indeed converges absolutely. For explicit computations of these periods in Euclidean $\phi^4$ theory in 4 dimensions, see for example \cite{Schnetz}. 

It is highly plausible that this result can also be generalized to other dimensions, e.g. $D=6$. For an elementary argument, first note that potential singularities of the integrand lie on $C\doteq X_\Gamma\cap \partial\Delta$,
the intersection of the hyper-surface $X_\Gamma\doteq \{\vec{\al}\in\RR^{|E(\Gamma)|}|\hat{\Psi}_\Gamma(\vec{\al})=0\}$ with the boundary $\partial\Delta$. If $C$ is just a collection of points, one can split the integration region into small neighborhoods of these points and the rest. For each such neighborhood one parametrizes the integral using spherical coordinates around the point and examines the behaviour of the integrand as the radius $r$ approaches 0. One can now observe that for each such integral, extra factors of $\al_e$ contribute $r^{(|E(\Gamma)|-1)(\frac{D}{2}-2)}$, the integration measure contributes $r^{|E(\Gamma)|-2}$, while the denominator contributes $r^{-(|V(\Gamma)|-2)\frac{D}{2}}$. The last assertion follows from the fact that $\hat{\Psi}_\Gamma$ is a degree $|V(\Gamma)|-1$ polynomial and because we are integrating over the simplex, the dominant contribution comes from degree $|V(\Gamma)|-2$ terms. Since $|V(\Gamma)|-1=\frac{D-2}{2}|E(\Gamma)|$, the integrand can be bounded by a constant, as $r\rightarrow 0$. We perform these estimates explicitly in Example~\ref{sixD}.

In proposition \ref{mainprop} we show how periods defined by \eqref{period} appear in distributional residues in Lorentzian signature. Before we do that, it is worth to recall a few facts concerning graph polynomials (see \cite{BognerThesis,BogWe} for a more comprehensive review).
\begin{df}[\cite{Sta98,Tut84}]\label{Lap}
	The generic graph Laplacian (or Kirchhoff matrix) is the $|V(\Gamma)|\times |V(\Gamma)|$ matrix defined by	
	\[
	\Lcal_{ij}(\vec{\al})=
	\begin{cases}
	\sum\limits_{e\in E(\Gamma)\atop v_i,v_j\in\partial e} -\al_e&\textrm{if}\ i\neq j,\\
	\sum\limits_{e\in E(\Gamma)\atop v_i\in\partial e}\al_e&\textrm{if}\ i=j,
	\end{cases}
	\]
	for all $v_i,v_j\in V(\Gamma)$. A sum over the empty set is set to be zero.
	\end{df}
\begin{thm}[tree-matrix theorem in \cite{Tut84}, thm. VI.29]\label{treematrix}
Let $\Gamma$ be a graph with $N$ edges, all of them labelled by the set $\{\al_1,...,\al_N\}$ and let $v_i$ be an arbitrary vertex of $\Gamma$. Let $\Lcal_\Gamma(\vec{\al})$ be the generic Laplacian and $\hat{\Psi}_\Gamma$ the dual graph
polynomial. Then we have
\[
\hat{\Psi}_\Gamma = \operatorname{Det} (\Lcal_\Gamma(\vec{\al})[v_i])\,,
\]
where the notation $\Lcal_\Gamma(\vec{\al})[v_i]$ means the $(i,i)$ minor of the matrix $\Lcal_\Gamma(\vec{\al})$.
\end{thm}
We are now ready to prove our main result of this section.
\begin{prop}\label{mainprop}
	Let $\Gamma$ be a graph with $|E(\Gamma)|=\frac{D}{2}h_1$ and such that every proper subgraph $\gamma$ satisfies $|E(\Gamma)|>2h_1$. If $P_\Gamma$ converges absolutely, then the distributional residue $\Res u^\Gamma$ is given by
	\[
	\Res u_\Gamma=
	c_0=\frac{2i^{(2D-1)(|V|-1)}}{(4\pi)^{|E(\Ga)|}}\,P_\Gamma\,.
	\]
\end{prop}
\begin{proof}
	First recall that the integral \eqref{c0} doesn't depend on the choice of $\Sigma$. The simplest choice is the unit Euclidean sphere in $\RR^{Dn}$, where $n=|V(\Gamma)|-1$.
	Denote
	\[
	X\equiv(x_1^0,\dots,x^0_n,\dots,x^{D-1}_1,\dots,x^{D-1}_n)
	\] 
%	The integral is invariant under re-parametrization $x\rightarrow -x$ so we can also write $c_0$ as twice the integral over half-sphere. The latter is then replaced by the integral on the projective space so
%	\[
%	c_0=2\int_{\mathbb{PR}^{N-1}}u(y)\Omega_p(y)\,,
%	\]  
%	where
%	\[
%	\Omega_p(y)=
%	\]
%	is the projective measure and $y$ are homogeneous coordinates. 
Using the formula \eqref{uGamma} we obtain
	\[
	c_0=(-1)^{(\frac{D}{2}-1)|E(\Ga)|}\left(\frac{\bGa(\frac{D}{2}-1)}{4\pi^{\frac{D}{2}}}\right)^{|E(\Ga)|}\lim_{\epsilon\rightarrow0^+}\int_{\Sigma}\frac{d\sigma(X)}{\prod_{e\in E(\Gamma)}((x_{s(e)}-x_{f(e)})^2-i\epsilon)^{\frac{D}{2}-1}}\,,
		\]
	 Denote $b_e\equiv (x_{s(e)}-x_{f(e)})^2-i\epsilon$, $e\in E(\Gamma)$. We have $\Re(ib_e)=\epsilon>0$, so we can use the well known Schwinger trick to write
	\begin{multline*}
	\frac{1}{\prod_{e\in E(\Gamma)}b_e^{\frac{D}{2}-1}}=\frac{\bGa((\frac{D}{2}-1)|E(\Gamma)|)}{(\bGa(\frac{D}{2}-1))^{|E(\Gamma)|}}\int_{0}^{1}\dots\int_{0}^{1} \frac{\delta(1-\sum_{e\in E(\Gamma)}\alpha_e)}{(\sum_{e\in E(\Gamma)} \alpha_eb_e)^{(\frac{D}{2}-1)|E(\Gamma)|}}\,\prod_{e\in E(\Gamma)}\alpha_e^{\frac{D}{2}-1}d\alpha_e\\\equiv \frac{\bGa((\frac{D}{2}-1)|E(\Gamma)|)}{(\bGa(\frac{D}{2}-1))^{|E(\Gamma)|}}\int\limits_\Delta \frac{\mu(\vec{\al})}{(\sum_{e\in E(\Gamma)} \alpha_eb_e)^{(\frac{D}{2}-1)|E(\Gamma)|}} \,,
	\end{multline*}
	where $k=|E(\Gamma)|$. Now we want to perform a change of variables to put the quadratic form $B\equiv \sum_{e\in E(\Gamma)} \alpha_eb_e$ into its normal form. We write $B=X^T M X$, where $M$ is a block diagonal matrix of the form
	\[
	M=\left(\begin{matrix}
	N&0&0&\dots&0\\
	0&-N&0&\dots&0\\
	0&0&-N&\dots&0\\
	\vdots&&&&	\vdots\\
	0&0&0&\dots&-N\\

	\end{matrix}\right)\,.
	\]
	Each block is a $(|V(\Gamma)|-1)$-dimensional symmetric positive semidefinite matrix (as $\alpha_e\geq0$, $\forall e\in E(\Gamma)$ ), which is in fact the $(0,0)$ minor of the generic graph Laplacian $\Lcal_\Gamma(\vec{\al})$ introduced in Definition \ref{Lap}. We can find a non-singular matrix $\Lambda$ such that
	\[
	\Lambda^T N\Lambda=\1\,.
	\]
	The argument proceeds now exactly the same as in \cite{BEK06,Bloch}. Defining
	\[
	S\doteq \left(\begin{matrix}
	\Lambda&0&0&\dots&0\\
	0&\Lambda&0&\dots&0\\
	0&0&\Lambda&\dots&0\\
	\vdots&&&&	\vdots\\
	0&0&0&\dots&\Lambda\\
	\end{matrix}\right)\,,
	\]
	we obtain
		\be\label{canon:form}
	S^TMS\doteq \left(\begin{matrix}
	\1&0&0&\dots&0\\
	0&-\1&0&\dots&0\\
	0&0&-\1&\dots&0\\
	\vdots&&&&	\vdots\\
	0&0&0&\dots&-\1\\
	\end{matrix}\right)\equiv \Xi\,,
	\ee
	This suggests a change of variables $X\mapsto S^{-1} X$ that puts the quadratic form $B$ into the normal form. In order to perform this change of variables we only need to ensure that in the following formula the order of integration can be interchanged: 
	\be\label{int:fub}
	\int\limits_{\Sigma}\int\limits_\Delta \frac{1}{(\sum_{e\in E(\Gamma)} \alpha_eb_e)^{(\frac{D}{2}-1)|E(\Gamma)|}}\,\mu(\vec{\al})d\sigma
	\ee
	For this, note that 
	\[
	|\sum_{e\in E(\Gamma)} \alpha_eb_e|^2\geq |\sum_{e\in E(\Gamma)} \al_e|^2\epsilon^2\,.
	\]
	Since on the simplex $\Delta$ we have $\sum_{e\in E(\Gamma)} \al_e=1$, we can conclude that
	the integrand in \eqref{int:fub} is uniformly bounded by
	 $ \frac{1}{\epsilon^2}$ and as long as $\epsilon>0$, we can interchange the order of integration and perform the desired change of variables $X\mapsto S^{-1} X$. The Jacobian for this change of variables is
	 \[
\operatorname{Det} S=(\operatorname{Det} \Lambda)^D=(\operatorname{Det} N)^{-D/2}\,,
	 \]
	 since $(\operatorname{Det} \Lambda)^2\operatorname{Det}N=1$.
	 It follows now from the tree-matrix theorem \ref{treematrix} that
	 \[
	 \operatorname{Det} N=\hat{\Psi}_{\Gamma}(\vec{\al})\,.
	 \]
	 It is now also explicitly seen that the result doesn't depend on the choice of the vertex to which we assigned 0 in our Feynman rules, as the tree-matrix theorem gives the same result for any choice of the minor $\Lcal_{\Gamma}[v_i]$, $v_i\in V(\Gamma)$. 
	 
	 We can now rewrite $c_0$ as
	\begin{multline}\label{coPsi}
	c_0=\bGa\left(|E(\Ga)|(\tfrac{D}{2}-1)\right)\left(\frac{(-1)^{(\frac{D}{2}-1)}}{4\pi^{\frac{D}{2}}}\right)^{|E(\Ga)|}\lim_{\epsilon\rightarrow 0^+}\int_{\Sigma}\frac{d\sigma}{(X^T\Xi X-i\epsilon)^{|E(\Ga)|(\frac{D}{2}-1)}}\int_{\Delta}\hat{\Psi}_\Gamma^{-D/2}(\vec{\al})\mu(\vec{\al})=\\
	\bGa\left(|E(\Ga)|(\tfrac{D}{2}-1)\right)\left(\frac{(-1)^{(\frac{D}{2}-1)}}{4\pi^{\frac{D}{2}}}\right)^{|E(\Ga)|}P_\Gamma\,\lim_{\epsilon\rightarrow 0^+}\int_{\Sigma}\frac{d\sigma}{(X^T\Xi X-i\epsilon)^{|E(\Ga)|(\frac{D}{2}-1)}}
	\,,
	\end{multline}
	where $\Xi$ is a diagonal metric given in \eqref{canon:form}.
	
	The remaining integral in \eqref{coPsi} is easy to evaluate. It is the residue of the distribution \[t(X)=\frac{1}{(X^T\Xi X-i0)^{(\frac{D}{2}-1)|E(\Ga)|}}\]
	on the indefinite product space $(\RR^{(D-2)|E(\Ga)|},\Xi)$, with divergence degree $0$ and scaling order 0. Now we use formula \cite[Appendix C, formula after eq.~(102)]{BDF}:
	\[
	\Res t=i^s|S^{d-1}|\,,
	\]
where $d$ is the total dimension of the indefinite product space (in our case $d=(|V(\Gamma)|-1)D=|E(\Gamma)|(D-2)$) and $s$ is the number of minus signs in the signature of $\Xi$ (in our case $s=(|V(\Gamma)|-1)(D-1)$) and $|S^{d-1}|$ is the volume of the unit sphere in $d$ dimensions. We obtain 
	\[
	\Res t=i^{(D-1)(|V(\Ga)|-1)}\left|S^{|E(\Ga)|(D-2)-1}\right|\,.
	\]
	With this result and the formula for the volume of the unit sphere in $d$ dimensions
	\[
	|S^{d-1}|=\frac{2\pi^{n/2}}{\bGa(\tfrac{n}{2})}\,,
	\]
	we arrive at
	\[
	c_0=i^{(2D-1)(|V(\Ga)|-1)}\frac{2}{(4\pi)^{|E(\Ga)|}}\,P_\Gamma\,.
	\]
\end{proof}
In particular, for $D=4$ we have
\[
c_0=(-i)^{(|V(\Ga)|-1)}\frac{2}{(4\pi)^{|E(\Ga)|}}\int_{\Delta}|\hat{\Psi}_\Gamma|^{-2}\Omega(\al) \,,
\]
where $\Omega(\al)$ is the standard measure on the simplex.
\begin{example}
The simplest example is the fish graph in 4 dimensions:
\begin{center}
	\begin{tikzpicture}[x=1.00mm, y=1.00mm, inner xsep=0pt, inner ysep=0pt, outer xsep=0pt, outer ysep=0pt]
	\path[line width=0mm] (82.33,85.75) rectangle +(20.09,8.50);
	\definecolor{L}{rgb}{0,0,0}
	\definecolor{F}{rgb}{0,0,0}
	\path[line width=0.30mm, draw=L, fill=F] (84.87,90.05) circle (0.54mm);
	\path[line width=0.30mm, draw=L, fill=F] (99.87,90.05) circle (0.54mm);
	\path[line width=0.30mm, draw=L] (85.00,90.00) .. controls (89.50,87.00) and (95.50,87.00) .. (100.00,90.00);
	\path[line width=0.30mm, draw=L] (85.00,90.00) .. controls (89.50,93.00) and (95.50,93.00) .. (100.00,90.00);
	\end{tikzpicture}%
\end{center}
 The scaling degree and the scaling order vanish, so from proposition \ref{mainprop} we obtain
\[
c_0=-i\frac{2}{(4\pi)^{2}}P_{\Gamma}\,.
\]
Here $\hat{\Psi}_\Gamma=\al_1+\al_2$, so $P_\Gamma=1$ and hence $c_0=\frac{-i}{8\pi^2}$.
\end{example}
\begin{example}\label{sixD}
Following \cite{BDF}, consider the triangle in 6 dimensions:
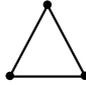
\begin{center}
	\begin{tikzpicture}[scale=0.5,x=1.00mm, y=1.00mm, inner xsep=0pt, inner ysep=0pt, outer xsep=0pt, outer ysep=0pt]
	\path[line width=0mm] (67.20,82.35) rectangle +(25.35,24.52);
	\definecolor{L}{rgb}{0,0,0}
	\definecolor{F}{rgb}{0,0,0}
	\path[line width=0.30mm, draw=L, fill=F] (79.87,104.05) circle (0.82mm);
	\path[line width=0.30mm, draw=L] (80.00,104.00);
	\path[line width=0.30mm, draw=L] (80.00,104.00) -- (70.00,85.00);
	\path[line width=0.30mm, draw=L] (80.02,103.68) -- (89.87,84.95);
	\path[line width=0.30mm, draw=L] (70.00,85.00) -- (89.50,85.00);
	\path[line width=0.30mm, draw=L, fill=F] (89.73,85.21) circle (0.82mm);
	\path[line width=0.30mm, draw=L, fill=F] (70.02,85.17) circle (0.82mm);
	\end{tikzpicture}
\end{center}
 Proposition \ref{mainprop} implies that
\[
c_0=-\frac{2}{(4\pi)^{3}} P_\Gamma\,,
\]
if $P_\Gamma$ converges. Since $\hat{\Psi}(\vec{\al})=\al_1\al_2+\al_1\al_3+\al_2\al_3$, we have
\[
P_\Gamma=
\int_{\Delta}\frac{\al_1\al_2\al_3\delta(1-\al_1-\al_2-\al_3)d\al_1d\al_2d\al_3}{(\al_1\al_2+\al_2\al_3+\al_1\al_3)^3}\,.
\]
To see that this integral is absolutely convergent, note that singularities of the integrand appear only in the ``corners'' of the simplex. Using the symmetry of the problem, we pick the $\alpha_3=1$ and consider the integral $I_\epsilon$ of the same integrand as above, but over a small neighborhood of the point $(\alpha_1,\alpha_2,1)$ on the simplex $\Delta$. Using polar coordinates $\alpha_1=r\cos\theta$, $\alpha_2=r\sin\theta$, this integral takes the form
\[
I_\epsilon=\int_0^\epsilon\int_0^{\pi/2}\frac{r^3\sin^2\theta\cos\theta(1-r\sqrt{2}\sin(\theta+\frac{\pi}{4}))}{r^3(\frac{1}{2}r\sin2\theta+\sqrt{2}\sin(\theta+\frac{\pi}{4})-2r\sin^2(\theta+\frac{\pi}{4}))^3}d\theta dr
\]
Since $\sin(\theta+\frac{\pi}{4})$ does not vanish in the interval $[0,\frac{\pi}{2}]$, the integrand can be bounded by a constant when $r\rightarrow 0$, so $I_\epsilon$ is absolutely convergent and so is $P_\Gamma$.

Following \cite[example on p.~39]{BDF} we evaluate this integral by integrating out $\al_3$ and then changing the variables to $\lambda,\kappa$, so that $\al_1=\lambda\kappa$ and $\al_2=(1-\lambda)\kappa$. We obtain
\[
P_\Gamma=\int_0^1\int_0^1 \frac{\lambda(1-\lambda)\kappa^2(1-\kappa)}{(\lambda(1-\lambda)\kappa^2+\kappa(1-\kappa))^3}d\kappa d\lambda=\frac{1}{2}\,,
\]
so
\[
c_0=-\frac{1}{2^6\pi^3}\,.\]
\end{example}
\begin{example}
The final example is the well known ``wheel with three spokes'' graph in 4 dimensions:
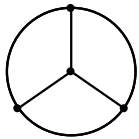
\begin{center}
	\begin{tikzpicture}[scale=0.7,x=1.00mm, y=1.00mm, inner xsep=0pt, inner ysep=0pt, outer xsep=0pt, outer ysep=0pt]
	\path[line width=0mm] (65.92,77.92) rectangle +(28.17,28.68);
	\definecolor{L}{rgb}{0,0,0}
	\definecolor{F}{rgb}{0,0,0}
	\path[line width=0.30mm, draw=L, fill=F] (79.87,104.05) circle (0.54mm);
	\path[line width=0.30mm, draw=L] (80.00,92.00) circle (12.08mm);
	\path[line width=0.30mm, draw=L] (80.00,92.00) circle (12.08mm);
	\path[line width=0.30mm, draw=L, fill=F] (89.87,85.05) circle (0.54mm);
	\path[line width=0.30mm, draw=L, fill=F] (69.87,85.05) circle (0.54mm);
	\path[line width=0.30mm, draw=L, fill=F] (79.87,92.05) circle (0.54mm);
	\path[line width=0.30mm, draw=L] (80.00,104.00);
	\path[line width=0.30mm, draw=L] (80.00,104.00) -- (80.00,92.00);
	\path[line width=0.30mm, draw=L] (70.00,85.00) -- (80.00,92.00);
	\path[line width=0.30mm, draw=L] (80.00,92.00) -- (90.00,85.00);
	\end{tikzpicture}%
\end{center}
 This one also satisfies the assumptions of proposition \ref{mainprop}, so using the general formula we obtain
\[
c_0=\frac{i}{2^{11}\pi^6} P_\Gamma=\frac{3i}{2^{10}\pi^6}\,\zeta(3)\,,
\]
where we used the well-known value $P_\Gamma=6\zeta(3)$ (see e.g. \cite{Brown09}).
\end{example}
Proposition \ref{mainprop} allows to reduce the problem of computing a large class of distributional residues to the problem of evaluating periods arising from graph polynomials, of the form discussed in \cite{Schnetz,BEK06,Brown,MaAl}, so can be used to easily translate the existing results and apply them to theories in Lorentzian signature.

Let us come back to the general case. Let $\Gamma$ be a graph with $\omega_\Gamma\geq 0$. If it contains proper subgraphs with  $\omega_\gamma\geq 0$, then one has to renormalize these first and substitute the result to the expression for $t^\Gamma$. If overlapping divergences are present, a partition of unity might be required. However, there are convincing arguments that this step can be avoided; compare the example 4.16 in \cite{DFKR} (using the partition of unity) with example 5.3 of \cite{GGV14} (without the partition of unity).
%\todo{ The crucial point here is that convolution over internal vertices	commutes with renormalization [8].} 
A distribution constructed this way is denoted by $\tilde{u}^{\Gamma}$ and it was shown in \cite{HW} that the property of almost homogeneous scaling is preserved in the recursive procedure of renormalization of proper subgraphs. Hence $\tilde{u}^{\Gamma}$ is an almost homogeneously scaling distribution and the general formula for its residue is
\[
\operatorname{Res}(\tilde{u}^\Gamma)=\sum_{|\beta|=\al-N}c_\beta\partial^\beta\,,
\]
where
\begin{equation}\label{gener}
c_\beta\doteq (-1)^{\al-N}\frac{1}{\beta!}\int_{\Sigma}\hat{x}^{\beta}{(\mathscr{E}+\al)^k\tilde{u}^\Gamma}(\hat{x})d\sigma(\hat{x})\,,
\end{equation}
If a graph is EG primitive, then $k=0$, $\tilde{u}^\Gamma=u^\Gamma$ and the residue is uniquely determined by the graph. Residues for EG primitive graphs which are not CK primitive can be obtained by using the fact that coefficients $c_\beta$ are Lorentz invariant. This implies that integrals \eqref{gener} can be reduced to scalar integrals multiplying appropriate powers of $\eta_{\mu\nu}$.

We believe that a result generalizing Proposition \ref{mainprop} can be established also in this case and we will address it in future work.
\begin{example}\label{sunset}
Consider the sunset diagram in 4 dimensions:
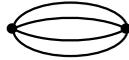
\begin{center}
	\begin{tikzpicture}[x=1.00mm, y=1.00mm, inner xsep=0pt, inner ysep=0pt, outer xsep=0pt, outer ysep=0pt]
	\path[line width=0mm] (82.33,84.40) rectangle +(20.09,11.07);
	\definecolor{L}{rgb}{0,0,0}
	\definecolor{F}{rgb}{0,0,0}
	\path[line width=0.30mm, draw=L, fill=F] (84.87,90.05) circle (0.54mm);
	\path[line width=0.30mm, draw=L, fill=F] (99.87,90.05) circle (0.54mm);
	\path[line width=0.30mm, draw=L] (85.00,90.00) .. controls (89.05,87.96) and (96.06,87.96) .. (100.00,90.00);
	\path[line width=0.30mm, draw=L] (85.00,90.00) .. controls (88.98,91.94) and (96.06,91.96) .. (100.00,90.00);
	\path[line width=0.30mm, draw=L] (92.56,89.94) ellipse (7.52mm and 3.54mm);
	\end{tikzpicture}
\end{center}
 We have $m=0$ and $\al=8$. This implies that $|\beta|=4$ so we need to compute
\[
c_{\mu\nu\al\beta}=\frac{1}{(2\pi)^84!}\int_{\Sigma}\frac{x^\mu x^\nu x^\al x^\beta}{(x^2-i0)^4}d\sigma(x)\,.
\]
The Lorenz invariance and the symmetry of the problem imply that
\begin{multline}
(2\pi)^8c_{\mu\nu\al\beta}=\frac{1}{4!24}(\eta_{\al\beta}\eta_{\mu\nu}+\eta_{\mu\beta}\eta_{\nu\alpha}+\eta_{\mu\al}\eta_{\nu\beta})\int_{\sigma}\frac{(x^2)^2}{(x^2-i0)^4}d\sigma(x)\\
=\frac{1}{2^63^2}(\eta_{\al\beta}\eta_{\mu\nu}+\eta_{\mu\beta}\eta_{\nu\alpha}+\eta_{\mu\al}\eta_{\nu\beta})\int_{\sigma}\frac{d\sigma(x)}{(x^2-i0)^2}\\=-\frac{i\pi^2}{2^53^2}(\eta_{\al\beta}\eta_{\mu\nu}+\eta_{\mu\beta}\eta_{\nu\alpha}+\eta_{\mu\al}\eta_{\nu\beta})
\end{multline}
Hence 
\[
\operatorname{Res}(u_\Gamma)=-\frac{i}{2^{13}3\pi^6}\Box^2\,.
\]
\end{example}
In fact there is a different, more direct, way to obtain residues for all the ``sunset'' type diagrams with arbitrary number of lines. For details see \cite[section 5.2]{NST} or \cite[Appendix C]{BDF}. The general formula is
\[
\Res\left(\frac{1}{(X^2-i0)^{\frac{d}{2}+l}}\right)=c_l\Box^l\,,
\]
where 
\[
c_l=i^s|S^{d-1}|\frac{\bGa(\frac{d}{2})}{2^{2l}l!\bGa(\frac{d}{2}+l)}
\]
and $X\in\RR^d$ with the diagonal metric of the form $\operatorname{diag}(1,\dots,1\underbrace{-\1,\dots,-\1}_{s})$. The example \ref{sunset} is then the special case of this formula with $d=4$, $s=3$ and $l=2$.
\section{Renormalization group flow}
 In \cite{BDF} the breaking of the homogeneous scaling is shown to relate to the definition of the $\beta$-function. In this section we review the main ideas of that argument.
 
 In the first step we generalize the discussion from the previous sections from the massless to the massive scalar field. For studying the scaling properties, it is crucial to work with time-ordered products that are smooth in mass\footnote{The usual physical argument for the 2-point functions not being smooth at $m^2=0$ is that it should not be possible to go smoothly to models with imaginary mass. However, the smoothness in mass is crucial for renormalization on curved spacetimes, as argued in \cite{HW01,HW02,BDF,HolWal08}. Another approach was proposed in \cite{Due15}, where the ``usual'' 2-point function can be used and the smoothness in mass is replaced by the smoothen of appropriately rescaled time-ordered products.}. This is, unfortunately, not the case if we use the standard Feynman propagator $\Delta^{\mathrm F}$. To rectify this, we replace in our framework the 2-point function $\Delta^+$ with a Hadamard 2-point function $H$ and the Feynman propagator $\Delta^{\mathrm F}$ with a corresponding modified Feynman propagator $H^{\mathrm{F}}$. Crucially, $H$ and $H^{\mathrm{F}}$ are smooth in mass. The choice of these objects is unique up to a parameter $M>0$ with the dimension of mass. Explicit formula for $H^{\mathrm{F}}_{\sst M}$ was derived in \cite{BDF} and it reads:
 \be\label{Feynmanprop}
\HFmu(x)=\frac{m^{D-2}}{(2\pi)^{\frac{D}{2}}\, y^{D-2}}\,\left(K_{\frac{D}{2}-1}(y)+(-1)^{\frac{D}{2}}\,
 \log{\frac{M}{m}}\ I_{\frac{D}{2}-1}(y)\right)\ ,
 \ee
 where $y\doteq\sqrt{-m^2(x^2-i0)}$ and $K$, $I$ are modified Bessel's functions. In 4 dimensions this amounts to
 %%%%%%%%%%%%%%%%%
 \begin{eqnarray}
 \HFmu(x)&=&\frac{-1}{4\pi^2(x^2-i0)}\notag\\
 &+&{\rm log}(-M^2(x^2-i0))\, m^2 f(m^2x^2)+m^2 F(m^2x^2)\ ,\notag
  \end{eqnarray}
  while in 6 dimensions
  \begin{eqnarray}
 \HFmu(x)&=&\frac{1}{4\pi^3(x^2-i0)^2}+\frac{m^2 f(m^2x^2)}{\pi\,(x^2-i0)}\notag\\
 &+&\frac{1}{\pi}\left({\rm log}(-M^2(x^2-i0))\, m^4 f'(m^2x^2)+m^4 F'(m^2x^2)\right)\ ,\notag
 \end{eqnarray}
 %%%%%%%%%%%%%%%%%%%%
 where $f$ and $F$ are real-valued analytic functions.
 $f$ and $f'$ can be expressed in terms of the Bessel functions $J_1$ and $J_2$, respectively, namely
 %%%%%%%%%%%%%
 \[
 f(z)\doteq \frac{1}{8\pi^2\sqrt{z}}\>J_1(\sqrt{z})\ ,\quad f(0)=\frac{1}{2^4\,\pi^2}\ ,
 \quad f'(z)=\frac{-1}{16\,\pi^2\,z}\>J_2(\sqrt{z})\ ;\label{f}
 \]
 %%%%%%%%%%%%%%%%%%%%
 and $F$ is given by a power series
 %%%%%%%%%%%%%
 \[
 F(z)\doteq -\frac{1}{4\pi}\sum_{k=0}^\infty\{\psi(k+1)+\psi(k+2)\}
 \frac{(-z/4)^k}{k!(k+1)!}\ , \quad F(0)=\frac{2\,C-1}{4\pi}\ ,
 \]
 %%%%%%%%%%%%%%%%%%%%
 where $C$ is Euler's constant and the Psi-function is related to the Gamma-function by
 $\psi(x)\doteq\bGa^\prime (x)\, /\, \bGa (x)$.
  
 The non-uniqueness of $H$ and $H^{\mathrm{F}}$ forces one to use a bit more abstract construction to define the observables and time-ordered product. 
 \begin{df}
 	For a mass $m$ we define a family of algebras $\fA(m)_{\sst M}\doteq (\Fcal_{\mc}[[\hbar]],\star_{\sst H})$, labeled by $M>0$, where $H\equiv H_{\sst M}^m$ and $\star_{\sst H}$ is defined by
 	\[
 	(F\star_{\sst H} G)(\ph)\doteq e^{\hbar\left\langle H,\frac{\delta^2}{\delta\ph\delta\ph'}\right\rangle}F(\ph)G(\ph')|_{\ph'=\ph}\,
 	\]
 \end{df}
 Different choices of the Hadamard 2-point function for a given mass $m$ differ by a smooth function, i.e.
$ H^m_{\sst M_1}-H^m_{\sst M_2}$ is smooth. This allows to define a homomorphism
\[
\al^m_{\sst M_1M_2}\doteq e^{\hbar\,\left\langle H^m_{\sst M_1}-H^m_{\sst M_2},\frac{\delta^2}{\delta\ph^2}\right\rangle}\,,
\]
between the algebras $\fA(m)_{\sst M_1}$ and $\fA(m)_{\sst M_2}$. We are now ready to define the algebra of observables for a fixed mass.
\begin{df}
	$\fA(m)$, the algebra of observables for mass $m$ consists of families $A=(A_{\sst M})_{M>0}$, where $A_{\sst M}\in\fA(m)_{\sst M}$ and we have $A_{\sst M_1}=\al^m_{\sst M_1M_2}(A_{\sst M_2})$.
\end{df}
We can identify abstract elements of the algebra $\fA(m)$ with concrete functionals in $\Fcal_{\mc}[[\hbar]]$. For $A\in\fA(m)$ denote
\[
A_{\sst M}\doteq \al_{\sst H}(A)\,, 
\]
where $\al_{\sst H}\equiv e^{ \left<\hbar H,\frac{\delta^2}{\delta\ph^2}\right>}$ and $H\equiv H_{\sst M}^m$ is the appropriate Hadamard 2-point function. $A_{\sst M}$ defined this way is now a functional in $\Fcal_{\mc}[[\hbar]]$. Conversely, let $F\in\Fcal_{\mc}$. We denote by
$\al_{\sst H}^{-1}F$ 
the element of $\fA(m)$ such that $(\al_{\sst H}^{-1}F)_{\sst M}=F$, where  $H\equiv H_{\sst M}^m$, as above. The rationale behind this notation is explained in \cite{BDF} and further clarified in \cite{Book}. Let $\fA_\loc(m)$ denote the subspace of $\fA(m)$ arising from local functionals.

Now, following \cite{BDF}, we want to combine algebras corresponding to different masses in a common algebraic structure.
\begin{df}
We define the following bundle of algebras
 \[
 \Bcal = \bigsqcup_{m^2\in\RR}\fA(m)\ .
 \]
\end{df} 
Let $A=(A^m)_{m^2\in\RR}$ be a section of $\Bcal$. We fix $M>0$ and define a function from $\RR_+$ to $\Fcal_{\mc}[[\hbar]]$ by
\[
m^2\mapsto \al_{\sst M}(A)(m)\doteq\al_{\sst H}(A^m)\,,\qquad \mathrm{where}\  H\equiv H_{\sst M}^m\,.
\]
\begin{df}
	A section $A$ of $\Bcal$ is called smooth if $\al_{\sst M}(A)$ is smooth for some (and hence all) $M>0$. The space of smooth sections of $\Bcal$ is denoted by $\fA$. Similarly, $\fA_\loc$ denotes the space of smooth sections of $\Bcal$ such that $A(m)\in\fA_\loc(m)$ for all $m$.
\end{df}
$\fA$ is equipped with a non-commutative product defined as follows:
\[
(A\star B)^m_{\sst M}\doteq A^m_{\sst M}\,\star_{\sst H}B^m_{\sst M}\,,
\]
where $H\equiv H_{\sst M}^m$. The $n$-fold time-ordered product $\Tcal^n$ is a map from $\fA_\loc$ to $\fA$ defined by
\[
\Tcal^n(A_1,\dots,A_n)(m)\doteq \al_{\sst H}^{-1}\circ \Tcal^n_{\sst H}(\al_{\sst H}A_1\dots,\al_{\sst H}A_n)\,,
\]
where $H\equiv H_{\sst M}^m$ is a Hadamard 2-point function for mass $m$ and maps $\Tcal^n_{\sst H}:\Fcal_{\loc}[[\hbar]]\rightarrow\Fcal_{\mc}[[\hbar]]$ satisfy axioms from Definition \ref{Tns1} with $\Delta_+$ replaced by $H$.

The $S$-matrix is now a map from $\fA_\loc$ to $\fA$ defined by
\[
\Scal(A)\doteq \sum_{n=0}^{\infty}\frac{1}{n!}\Tcal^n(A^{\otimes n})\,.
\]
Axioms for time-ordered products can be conveniently formulated on the level of $S$-matrices.
\begin{enumerate}[{\bf S 1.}]
	\item {\bf Causality} $\Scal(A+B)=\Scal(A)\star \Scal(B)\ $ if $\supp(A^m)$ is later than $\supp(B^m)$ for all $m^2\in\RR_+$.\footnote{We define $\supp A^m\doteq \supp (\alpha_{\sst H}(A))$, where $H\equiv H_{\sst M}^m$ and this definition is independent of the choice of $M$.}
	\item $\Scal(0)=1$, $\Scal^{(1)}(0)=\id$,
	\item {\bf $\ph$-Locality}: $\al_{\sst M}\circ \Scal(A)(\ph_0)=\al_{\sst M}\circ \Scal\circ \al_{\sst M}^{-1}\left(\al_{\sst M}(A)_{\ph_0}^{[N]}\right)(\ph_0)+O(\hbar^{N+1})$, where
	\[
	\al_{\sst M}(A)_{\ph_0}^{[N]}(\ph)=\sum_{n=0}^{N}\frac{1}{n!}\left\langle\frac{\delta^n\al_{\sst M}(A)}{\delta\ph^n}(\ph_0),(\ph-\ph_0)^{\otimes n}\right\rangle
	\]
	is the Taylor expansion up to order $N$. The dependence on mass $m$ is kept implicit in all these formulas.
	\item {\bf Field independence}: $\Scal$ doesn't explicitly depends on field configurations.
%	\item {\bf Unitarity}: $\overline{\Scal(\overline{A})}\star \Scal(A)=1$, where the bar denotes complex conjugation.	\todo{Maybe add unitarity earlier as well...}
\end{enumerate}
 In Epstein-Glaser renormalization the freedom in defining the renormalized S-matrix is controlled by the  St{\"u}ckelberg-Petermann renormalization group. 
 \begin{df}
 	The St{\"u}ckelberg-Petermann renormalization group $\Rcal$ is defined as the group of maps $Z:\fA_\loc\rightarrow\fA_\loc$ with the following properties:
 	\begin{enumerate}[{\bf Z 1.}]
 		\item $Z(0)=0$,\label{Zs}
 		\item $Z^{(1)}(0)=\id$,
 		\item $Z=\id+\Ocal(\hbar)$,
 		\item $Z(F+G+H)=Z(F+G)+Z(G+H)-Z(G)$, if $\supp\,F\cap\supp\,G=\emptyset$,
 		\item $\frac{\delta Z}{\delta\ph}=0$.\label{Zf}
 	\end{enumerate}
 \end{df}
Note that constructing $Z$'s can be reduced to constructing maps $Z_{\sst H}:\Fcal_{\loc}[[\hbar]]\rightarrow\Fcal_{\loc}[[\hbar]]$ which control the freedom in constructing $\Tcal^n_{\sst H}$, so the abstract formalism reviewed in the present section can be related to the more concrete description presented in sections 1-3. We have
\[
Z=\al_{\sst H}^{-1}\circ Z_{\sst H}\circ \al_{\sst H}\,.
\]
The fundamental result in the Epstein-Glaser approach to renormalization is \textit{the Main Theorem of Renormalization} (\cite{PS82,Stora02,DF04,BDF}.
\begin{thm}
	Given two $S$-matrices $\Scal$ and $\widehat{\Scal}$ 
	satisfying conditions {\bf S 1}-- {\bf S 5}, there exists a unique $Z\in\mathcal{R}$ such that 
	\begin{equation}
	\widehat{\Scal}=\Scal\circ Z \ .\label{mainthm}
	\end{equation}
	Conversely, given an $S$-matrix $\Scal$ satisfying the 
	mentioned conditions and a $Z\in\mathcal{R}$, equation (\ref{mainthm})
	defines a new $S$-matrix $\widehat{\Scal}$ satisfying {\bf S 1}-- {\bf S 5}.
\end{thm}
Let us now discuss symmetries. Again, we follow closely \cite{BDF}. Let $G$ be a subgroup of the automorphism group of $\fA$. Assume that it has a well defined action on $\mathscr{S}$, the space of S-matrices, by
\[
\Scal\mapsto g\circ\Scal\circ g^{-1}\,,
\]
where $\Scal\in\mathscr{S}$, $g\in G$. Since $g\circ\Scal\circ g^{-1}\in\mathscr{S}$, it follows from the Main Theorem of Renormalization that there exists an element $Z(g) \in\mathcal{R}$ such that  
\[
g\circ\Scal\circ g^{-1}=\Scal\circ Z(g)\,.
\]
 We obtain a cocycle in $\mathcal{R}$,
 \begin{equation}
 Z(gh)=Z(g)gZ(h) g^{-1} \ .
 \end{equation}
 The cocycle can be trivialized, i.e. is a coboundary, if there exists an element $Z\in\mathcal{R}$ such that
 \begin{equation}
 Z(g)=Z g Z^{-1}g^{-1} \quad\quad\forall g\in G \,.
 \end{equation}
If this is the case, then
\[
g\circ\Scal\circ g^{-1}=\Scal\circ Z g Z^{-1}g^{-1}\,.
\]
Hence
\[
g\circ\Scal\circ Z\circ g^{-1}=\Scal\circ Z \,,
\]
so the $S$-matrix $\Scal\circ Z$ is $G$-invariant.

The non-triviality of the cocycle corresponds to the existence of anomalies. One of the most prominent examples where the cocycle cannot be trivialized is the action of the scaling transformations.

The scaling transformation is defined first on the level of field configurations $\ph\in\Ecal$ as
 \begin{equation}
(\sigma_\rho \ph)(x)= \rho^{\frac{2-D}{2}}\ph(\rho^{-1} x)\,,\label{def-sigma}
\end{equation}
where $D$ is the dimension of $\MM$. This induces the action on functionals by the pullback $\si_\rho(F)(\ph)\doteq F(\si_\rho(\ph))$ and finally, the action on $\fA$ can be defined by
\[
\si_{\rho}(A)^m=\si_{\rho}(A^{\rho^{-1}m}) \,.
\]
Let now
\begin{equation}
\si_{\rho}\circ S\circ \si_{\rho}^{-1}=S \circ Z(\rho) \,.\label{scalingS}
\end{equation}
$Z(\rho)$ is called the \textit{Gell-Mann Low cocycle} and it satisfies the cocycle condition
\begin{equation}\label{cocycle}
Z(\rho_1\rho_2)=Z(\rho_1)\sigma_{\rho_1} Z(\rho_2)\sigma_{\rho_1}^{-1} \,.
\end{equation}
Typically this cocycle cannot be trivialized. The generator of this cocycle, denoted by $B$ is related to the $\beta$-function known from the physics literature. Following \cite{BDF} we define
\begin{equation}
B\doteq\rho\frac{d}{d\rho}Z(\rho)\Big\vert_{\rho=1} \,,\label{def(B)}
\end{equation} 
The physical $\beta$-function can be obtained from $B$ after one corrects for the ``wave function renormalization'' and ``mass renormalization'' (see \cite[section 6.4]{BDF} for details).

To find $B$ we differentiate \eqref{scalingS} and obtain
\[
\rho\frac{d}{d\rho}(\sigma_{\rho}\circ \Scal\circ\sigma_{\rho}^{-1})(V)\Big\vert_{\rho=1}=\rho\frac{d}{d\rho}(S\circ Z(\rho))(V)\Big\vert_{\rho=1}=\left<S^{(1)}(V),B(V)\right> \ ,
\]
Note that $\left<S^{(1)}(V),.\right>$ is invertible in the sense of formal power series so
\[
B(V)=S^{(1)}(V)^{-1}\circ\rho\frac{d}{d\rho}(\sigma_{\rho}\circ \Scal\circ\sigma_{\rho}^{-1})(V)\Big\vert_{\rho=1}
\]
To compute $B$, first we write it in terms of its Taylor expansion:
\begin{equation}
 B(V)=\sum_{n=0}^\infty\frac{1}{n!}\, \left<B^{(n)}(0),V^{\otimes n}\right>\ ,\label{expansionB}
\end{equation}
where
\[
\left<B^{(n)}(0),V^{\otimes n}\right>=\left. \frac{d^n}{d\lambda^n} B(\lambda V)\right |_{\lambda=0}=\left. \rho \frac{d}{d\rho}\left. \frac{d^n}{d\lambda^n}\, Z(\rho)(\lambda V)\right |_{\lambda=0,\rho=1}\right.
\]
Denote $B^{(n)}(0)\equiv B^{(n)}$. The computation of $B^{(n)}$ amounts to summing up the scaling violations of distributional extensions appearing at order $n$ in construction of time-ordered products. To see that lower orders do not contribute, we use the fact that
\be\label{diffZS}
Z(\rho)^{(n)}(0)=\sigma_\rho\circ \Scal^{(n)}(0)\circ \sigma_{\rho}^{-1}-(\Scal\circ Z_{n-1}(\rho))^{(n)}(0)\,,
\ee
where $Z_n$ is an element of $\Rcal$ defined in terms of its Taylor expansion as
\begin{align}\label{Z-subdiag}
Z_n^{(k)}(0) &\doteq
\begin{cases}
Z^{(k)}(0)\ , & k\le n\ , \\
0\ , & k > n\ .
\end{cases}
\end{align}
The proof of \eqref{diffZS} is provided in \cite{BDF} and relies on the proof of the Main Theorem of Renormalization (Theorem 4.1 in \cite{BDF}). We expand $Z(\rho)^{(n)}(0)$ in terms of Feynman graphs:
\[
Z(\rho)^{(n)}(0)=\sum_{\Gamma\in\Gcal_n} Z(\rho)^{\Gamma}\,.
\]
where the sum is over all graphs with $n$ vertices. Similarly for $\Scal^{(n)}(0)$ and $B^{(n)}(0)$. We can rewrite \eqref{diffZS} as
\be\label{graphZS}
Z(\rho)^\Gamma=\sigma_\rho\circ \Tcal^\Gamma\circ \sigma_{\rho}^{-1}-\sum_{P\in \mathrm{Part}'(V(\Gamma))}\Tcal^{\Gamma_P}\circ\bigotimes_{I\in P} Z(\rho)^{\Gamma_I} \,,
\ee
where $\mathrm{Part}'(V(\Gamma))$ denotes the set of partitions of the vertex set $V(\Gamma)$, excluding the partition with $n$ elements; $\Gamma_P$ is the graph with vertex set $V(\Gamma_P)=V(\Gamma)$, with all lines connecting different index sets of the partition $P$, and $\Gamma_I$ is the graph with vertex set $V(\Gamma_I)=I$ and all lines of $\Gamma$ which connect two vertices in $I$. Differentiating \eqref{graphZS} with respect ot $\rho$ gives
\be\label{graphB}
B^\Gamma=
\rho\frac{d}{d\rho}(\sigma_\rho\circ \Tcal^\Gamma\circ \sigma_{\rho}^{-1})\Big\vert_{\rho=1}-\sum_{P\in \mathrm{Part}'(V(\Gamma))}\Tcal^{\Gamma_P}\circ\bigotimes_{I\in P} B^{\Gamma_I} \,,
\ee
Note that $B^\Gamma$ is an operator on $\Fcal_{\loc}^{\otimes n}[[\hbar]]$.

It is now clear that the second term in \eqref{graphZS} subtracts contributions from scaling violations corresponding to renormalization of all proper subgraphs of $\Gamma$. Hence the only contributions to $B^\Gamma$ arise from scaling violations resulting from extending distributions defined everywhere outside the thin diagonal of the graph $\Gamma$.

For performing computations we need to express $V\in\fA$ in terms of a concrete functional in $\Fcal_\loc$. Let's take $V=\al_{\sst M}^{-1} F$ for some $F\in\Floc$.
In the computation of $B$ we have to take into account that $\alpha_{\sst M}$,  does not commute with the scaling transformations. Define
\[\Scal_{\sst M}\doteq\al_{\sst M}\circ \Scal\circ \al_{\sst M}^{-1}
\]
 and 
 \[
 B_{\sst M}\doteq\al_{\sst M}\circ B\circ \al_{\sst M}^{-1}\]
We obtain 
\begin{multline*}
\rho\frac{\partial}{\partial\rho}(\sigma_{\rho}\circ \Scal_{\sst M}\circ\sigma_{\rho}^{-1})(F)-
M\frac{\partial}{\partial M}\Scal_{\sst M}(F)\Big \vert_{\rho=1}= \rho\frac{d}{d\rho}(\sigma_{\rho}\circ \Scal_{\rho^{-1}\sst M}\circ\sigma_{\rho}^{-1})(F)\Big \vert_{\rho=1}\\
=
\left<\Scal_{\sst M}^{(1)}(F),B_{\sst M}(F)\right>\,. 
\end{multline*}
for $V\in\mathcal{F}_{loc}$. The expression for $-
M\frac{\partial}{\partial M}\Scal_{\sst M}$ was derived in \cite{BDF} and is given by 
\[
M\frac{\partial}{\partial M}S_{\sst M}^{(n)}=2\hbar\,\, S_{\sst M}^{(n)}\circ  \sum_{i\ne j}D_v^{ij}\,,
\]
where  $D_v^{ij}\doteq\frac12\left< v,\frac{\delta^2}{\delta\varphi_i\delta\varphi_j}\right>$ is a 
functional differential operator on $\mathcal{F}_\loc^{\otimes n}$ and $v\doteq \frac{1}{2}M\frac{d}{dM} H_{\sst M}^m$.

Again, $B_{\sst M}$ can be written in terms of its Taylor expansion and $B_{\sst M}^{(n)}(0)$ is expressed as a sum over graphs with $n$ vertices. Finally, note that due to the field independence of $\Scal$ and $Z$, we have

\[
\frac{\delta^n}{\delta\ph^n}\circ B_{\sst M}(F)=\sum_{P\in\operatorname{Part}(n)}
B_{\sst M}^{(|P|)}\circ\bigotimes_{I\in P} F^{|I|}\ .
\] 
It follows that the Taylor expansion of $B_{\sst M}(F)$ around $\ph=0$ is determined by the values of $B_{\sst M}^{(k)}(F^{(n_1)}\otimes\cdots\otimes F^{(n_k)})$ at $\ph=0$, where $n_1+\dots+n_k=n$. We will see now that this allows to express everything in terms of connected graphs.

Let $F\in\Floc$. Without loss of generality we can assume $F$ to be monomial, i.e. of the form
\be\label{localF}
F(\ph)=\int_{\MM} f(x)p(j_x(\ph))d^Dx\,,
\ee
where $f\in\Dcal$ and $p$ is a monomial function on the jet space and $j_x(\ph)$ is a finite order jet of $\ph$ at point $x$. Graphically, we can represent $F$ as a vertex, decorated by $f$ with one external leg for each factor of $\ph$, some of them carrying derivatives. For example $\int_{\MM} f(x)\ph^4(x)d^Dx$ is
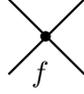
\begin{center}
	\begin{tikzpicture}[x=1.00mm, y=1.00mm, inner xsep=0pt, inner ysep=0pt, outer xsep=0pt, outer ysep=0pt]
	\path[line width=0mm] (83.00,80.92) rectangle +(18.67,16.08);
	\definecolor{L}{rgb}{0,0,0}
	\definecolor{F}{rgb}{0,0,0}
	\path[line width=0.30mm, draw=L, fill=F] (89.87,90.05) circle (0.54mm);
	\path[line width=0.30mm, draw=L] (95.00,95.00) -- (85.00,85.00);
	\path[line width=0.30mm, draw=L] (85.00,95.00) -- (95.00,85.00);
	\draw(88.00,84.00) node[anchor=base west]{\fontsize{11.23}{17.07}\selectfont $f$};
	\end{tikzpicture}
\end{center}
Given a monomial $p$ on the jet space, define the set of Wick submonomials $W_p$ as the set of all monomials that are factors of $p$. For example, for $\ph^4(x)$, the set of Wick submonomials consists of  $\ph^4(x)$, $\ph^3(x)$,  $\ph^2(x)$,  $\ph(x)$, 1. To indicate derivatives, we put lines across edges, e.g.
 $p(j_x(\ph))=\partial_\mu\ph\partial_\nu\ph$ is
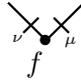
\begin{center}
	\begin{tikzpicture}[x=1.00mm, y=1.00mm, inner xsep=0pt, inner ysep=0pt, outer xsep=0pt, outer ysep=0pt]
	\useasboundingbox  (86.00,88.92) rectangle +(7,5.08);
	\path[line width=0mm] (81.54,83.02) rectangle +(16.96,13.98);
	\definecolor{L}{rgb}{0,0,0}
	\definecolor{F}{rgb}{0,0,0}
	\path[line width=0.30mm, draw=L, fill=F] (89.87,90.05) circle (0.54mm);
	\path[line width=0.30mm, draw=L] (95.00,95.00) -- (90.00,90.00);
	\path[line width=0.30mm, draw=L] (85.00,95.00) -- (90.00,90.00);
	\draw(87.32,85.89) node[anchor=base west]{\fontsize{11.38}{13.66}\selectfont $f$};
	\path[line width=0.30mm, draw=L] (87.00,91.00);
	\path[line width=0.30mm, draw=L] (87.00,91.00) -- (89.00,93.00);
	\path[line width=0.30mm, draw=L] (91.00,93.00) -- (93.00,91.00);
	\draw(92.39,89.59) node[anchor=base west]{\fontsize{6}{6.83}\selectfont $\mu$};
	\draw(85.54,89.75) node[anchor=base west]{\fontsize{6}{6.83}\selectfont $\nu$};
	\end{tikzpicture}
\end{center}
and after summing up over the index $\mu$ we obtain $\partial_\mu\ph\partial^{\mu}\ph\equiv (\partial \ph)^2$ represented for simplicity by
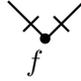
\begin{center}
	\begin{tikzpicture}[x=1.00mm, y=1.00mm, inner xsep=0pt, inner ysep=0pt, outer xsep=0pt, outer ysep=0pt]
	\useasboundingbox  (86.00,88.92) rectangle +(7,5.08);
	\path[line width=0mm] (81.54,83.02) rectangle +(16.96,13.98);
	\definecolor{L}{rgb}{0,0,0}
	\definecolor{F}{rgb}{0,0,0}
	\path[line width=0.30mm, draw=L, fill=F] (89.87,90.05) circle (0.54mm);
	\path[line width=0.30mm, draw=L] (95.00,95.00) -- (90.00,90.00);
	\path[line width=0.30mm, draw=L] (85.00,95.00) -- (90.00,90.00);
	\draw(87.32,85.89) node[anchor=base west]{\fontsize{11.38}{13.66}\selectfont $f$};
	\path[line width=0.30mm, draw=L] (87.00,91.00);
	\path[line width=0.30mm, draw=L] (87.00,91.00) -- (89.00,93.00);
	\path[line width=0.30mm, draw=L] (91.00,93.00) -- (93.00,91.00);
	\end{tikzpicture}
\end{center}
The Taylor expansion induces a coproduct
\[
p(j_x(\ph+\psi))=\Delta(p)(j_x(\ph)\otimes j_x(\psi))\,,
\]
which can be written explicitly as
\[
\Delta(p)=\sum_{q\in W_p} \operatorname{Sym}(q)\, p/q\otimes q\,,
\]
where $p/q$ is the graph obtained by removing the edges corresponding to $q$ and $\operatorname{Sym}(q)$ is the number of ways in which graph $q$ can be embedded into graph $p$.
For the local functional $F$ in \eqref{localF} we obtain
\[
F(\ph+\psi)=\int_{\MM}f(x)\Delta p(j_x(\ph)\otimes j_x(\psi)) d^Dx\,.
\]
Using Sweedler's notation:
\[
\Delta p=\sum_p p_{(1)}\otimes p_{(2)}\,.
\]
By a small abuse of notation, we define a functional $F_{(1)}(\ph)\doteq \int_{\MM} f(x)p_{(1)}(j_x(\ph))d^Dx$, while $F_{(2)}(\ph)(x)$ is a smooth function defined by $x\mapsto p_{(2)}(j_x(\ph))$. Using this notation:
\[
B^{(n)}_{\sst M}(F_1,\dots,F_n)(\ph)=\sum_{F_1,\dots,F_n}\left<B^{(n)}_{\sst M}({F_1}_{(1)},\dots,{F_n}_{(1)})(0),{F_1}_{(2)},\dots,{F_n}_{(2)}\right>\,.
\]
Here $B^{(n)}_{\sst M}({F_1}_{(1)},\dots,{F_n}_{(1)})(0)$ is a distribution, which we can write as
\[
B^{(n)}_{\sst M}({F_1}_{(1)},\dots,{F_n}_{(1)})(0)(x_1,\dots,x_n)=f_1(x_1)\dots f_n(x_n) \sum_{\Gamma} b^{\Gamma}(x_1,\dots,x_n)\,,
\]
where the sum runs over connected graphs $\Gamma$ with vertices representing ${p_1}_{(1)},\dots, {p_n}_{(1)}$. Distributions $b^\Gamma$ are given by
\[
b^\Gamma=\rho \frac{d}{d\rho}\sigma_\rho(\overline{u}^{\Gamma})\Big\vert_{\rho=1}\,,
\]
where $\overline{u}^{\Gamma}$ is the extension to the total diagonal of the distribution $\tilde{u}^\Gamma$ constructed as in section \ref{resid}, where all the proper subgraphs have been renormalized. Hence
\be
B^{(n)}_{\sst M}(F_1,\dots,F_n)(\ph)=\sum_{F_1,\dots,F_n}\sum_{\Gamma}\left< (f_1\otimes\dots\otimes f_n)\cdot b^{\Gamma},{F_1}_{(2)},\dots,{F_n}_{(2)}\right>\,.
\ee

If $\Gamma$ is EG primitive, then $\tilde{u}^\Gamma=u^{\Gamma}$ and $u^{\Gamma}$ scales homogeneously. In this case
\[
b^\Gamma=\Res u^\Gamma\,.
\]
This result provides a link between Kontsevich-Zagier periods appearing in Proposition~\ref{mainprop} and physical quantities computed in the pAQFT framework.
 However, the class of distributional residues relevant for the computation of $B$ is larger than the ones discussed in section \ref{resid}, since here we need to replace $\DF$ with $H^{\mathrm{F}}$ given by the formula \eqref{Feynmanprop}. To give an idea of how the computation proceeds at low loop orders, we review the example of $\ph^4$ in 4 dimensions discussed in \cite{BDF}, but in contrast to \cite{BDF} we use the Feynman graphs notation to make it easier to follow.
\begin{example}
Consider the functional
\[
F(\ph)=\lambda\int_{\MM}f(x)\ph^4(x)d^4x\,.
\]
The corresponding element of $\fA$ is
\[
V=\alpha^{-1}_{\sst M} F\,,
\]
i.e. 
\[
V(m)_{\sst M}=\lambda\, \al^{-1}_{H_{\sst M}^m}\int_{\MM}f(x)\ph^4(x)d^4x\,.
\]
We are interested in finding $B_{\sst M}$ for the QFT model with this interaction. First note that the orbit of the renormalization group is spanned by $1$ and functionals of the form
$\int_\MM f_1(x)\ph^4(x)d^4x$, $\int_\MM f_2(x)\ph^2(x)d^4x$, $\int_\MM f_3(x)(\partial\ph)^2(x)d^4x$, where $f_1,f_2,f_3\in\Dcal$. Hence, we need to determine $B_{\sst M}$ only on such functionals. Graphically we represent them as decorated vertices:
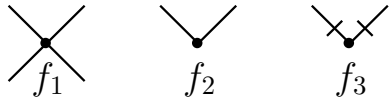
\begin{center}
\begin{tikzpicture}[x=1.00mm, y=1.00mm, inner xsep=0pt, inner ysep=0pt, outer xsep=0pt, outer ysep=0pt]
\path[line width=0mm] (83.00,80.92) rectangle +(18.67,16.08);
\definecolor{L}{rgb}{0,0,0}
\definecolor{F}{rgb}{0,0,0}
\path[line width=0.30mm, draw=L, fill=F] (89.87,90.05) circle (0.54mm);
\path[line width=0.30mm, draw=L] (95.00,95.00) -- (85.00,85.00);
\path[line width=0.30mm, draw=L] (85.00,95.00) -- (95.00,85.00);
\draw(88.00,84.00) node[anchor=base west]{\fontsize{14.23}{17.07}\selectfont $f_1$};
\end{tikzpicture}
\begin{tikzpicture}[x=1.00mm, y=1.00mm, inner xsep=0pt, inner ysep=0pt, outer xsep=0pt, outer ysep=0pt]
\path[line width=0mm] (83.00,80.92) rectangle +(18.67,16.08);
\definecolor{L}{rgb}{0,0,0}
\definecolor{F}{rgb}{0,0,0}
\path[line width=0.30mm, draw=L, fill=F] (89.87,90.05) circle (0.54mm);
\path[line width=0.30mm, draw=L] (95.00,95.00) -- (90.00,90.00);
\path[line width=0.30mm, draw=L] (85.00,95.00) -- (90.00,90.00);
\draw(88.00,84.00) node[anchor=base west]{\fontsize{14.23}{17.07}\selectfont $f_2$};
\end{tikzpicture} 
\begin{tikzpicture}[x=1.00mm, y=1.00mm, inner xsep=0pt, inner ysep=0pt, outer xsep=0pt, outer ysep=0pt]
\path[line width=0mm] (83.00,80.92) rectangle +(18.67,16.08);
\definecolor{L}{rgb}{0,0,0}
\definecolor{F}{rgb}{0,0,0}
\path[line width=0.30mm, draw=L, fill=F] (89.87,90.05) circle (0.54mm);
\path[line width=0.30mm, draw=L] (95.00,95.00) -- (90.00,90.00);
\path[line width=0.30mm, draw=L] (85.00,95.00) -- (90.00,90.00);
\draw(88.00,84.00) node[anchor=base west]{\fontsize{14.23}{17.07}\selectfont $f_3$};
\path[line width=0.30mm, draw=L] (87.00,91.00);
\path[line width=0.30mm, draw=L] (87.00,91.00) -- (89.00,93.00);
\path[line width=0.30mm, draw=L] (91.00,93.00) -- (93.00,91.00);
\end{tikzpicture}
\end{center}
Let us now compute $B^{(2)}_{\sst M}$ on these functionals. We have
\begin{multline}\label{B2M}
B^{(2)}_{\sst M}\left(\begin{tikzpicture}[scale=0.7,x=1.00mm, y=1.00mm, inner xsep=0pt, inner ysep=0pt, outer xsep=0pt, outer ysep=0pt]%, baseline=0pt]
\useasboundingbox  (83.00,88.92) rectangle +(14.67,5.08);
\definecolor{L}{rgb}{0,0,0}
\definecolor{F}{rgb}{0,0,0}
%\path[line width=1mm, draw=L,anchor=base] (83.00,80.92) rectangle +(10.67,10.08);
\path[line width=0.30mm, draw=L, fill=F, anchor=base] (89.87,90.05) circle (0.54mm);
\path[line width=0.30mm, draw=L] (95.00,95.00) -- (85.00,85.00);
\path[line width=0.30mm, draw=L] (85.00,95.00) -- (95.00,85.00);
\draw(87.00,84.00) node[anchor=base west]{\fontsize{11.23}{17.07}\selectfont $\footnotesize{f_1}$};
\end{tikzpicture},%
\begin{tikzpicture}[scale=0.7,x=1.00mm, y=1.00mm, inner xsep=0pt, inner ysep=0pt, outer xsep=0pt, outer ysep=0pt]%, baseline=0pt]
\useasboundingbox  (83.00,88.92) rectangle +(14.67,5.08);
\definecolor{L}{rgb}{0,0,0}
\definecolor{F}{rgb}{0,0,0}
%\path[line width=1mm, draw=L,anchor=base] (83.00,80.92) rectangle +(10.67,10.08);
\path[line width=0.30mm, draw=L, fill=F, anchor=base] (89.87,90.05) circle (0.54mm);
\path[line width=0.30mm, draw=L] (95.00,95.00) -- (85.00,85.00);
\path[line width=0.30mm, draw=L] (85.00,95.00) -- (95.00,85.00);
\draw(87.00,84.00) node[anchor=base west]{\fontsize{11.23}{17.07}\selectfont $\footnotesize{f_1}$};
\end{tikzpicture}
\right)=16\left<B^{(2)}_{\sst M}\left(%
\begin{tikzpicture}[scale=0.7,x=1.00mm, y=1.00mm, inner xsep=0pt, inner ysep=0pt, outer xsep=0pt, outer ysep=0pt]%, baseline=0pt]
\useasboundingbox  (83.00,88.92) rectangle +(14.67,5.08);
\definecolor{L}{rgb}{0,0,0}
\definecolor{F}{rgb}{0,0,0}
\path[line width=0.30mm, draw=L, fill=F] (89.87,90.05) circle (0.54mm);
\path[line width=0.30mm, draw=L] (95.00,95.00) -- (90.00,90.00);
\path[line width=0.30mm, draw=L] (85.00,95.00) -- (90.00,90.00);
\draw(87.00,84.00) node[anchor=base west]{\fontsize{11.23}{17.07}\selectfont $\scriptsize{f_1}$};
\path[line width=0.30mm, draw=L] (90.00,95.00) -- (90.00,90.50);
\end{tikzpicture},%
\begin{tikzpicture}[scale=0.7,x=1.00mm, y=1.00mm, inner xsep=0pt, inner ysep=0pt, outer xsep=0pt, outer ysep=0pt]%, baseline=0pt]
\useasboundingbox  (83.00,88.92) rectangle +(14.67,5.08);
\definecolor{L}{rgb}{0,0,0}
\definecolor{F}{rgb}{0,0,0}
\path[line width=0.30mm, draw=L, fill=F] (89.87,90.05) circle (0.54mm);
\path[line width=0.30mm, draw=L] (95.00,95.00) -- (90.00,90.00);
\path[line width=0.30mm, draw=L] (85.00,95.00) -- (90.00,90.00);
\draw(87.00,84.00) node[anchor=base west]{\fontsize{11.23}{17.07}\selectfont $\scriptsize{f_1}$};
\path[line width=0.30mm, draw=L] (90.00,95.00) -- (90.00,90.50);
\end{tikzpicture}
\right)(0),\begin{tikzpicture}[scale=0.7,x=1.00mm, y=1.00mm, inner xsep=0pt, inner ysep=0pt, outer xsep=0pt, outer ysep=0pt]
\useasboundingbox  (86.00,88.92) rectangle +(7,5.08);
%\path[line width=0mm] (85.00,81.88) rectangle +(13.33,15.12);
\definecolor{L}{rgb}{0,0,0}
\definecolor{F}{rgb}{0,0,0}
\path[line width=0.30mm, draw=L, fill=F] (89.96,89.99) circle (0.54mm);
%\draw(87.00,84.00) node[anchor=base west]{\fontsize{11.23}{13.66}\selectfont $f_1$};
\path[line width=0.30mm, draw=L] (90.00,95.00) -- (90.00,90.50);
\end{tikzpicture}\otimes%
\begin{tikzpicture}[scale=0.7,x=1.00mm, y=1.00mm, inner xsep=0pt, inner ysep=0pt, outer xsep=0pt, outer ysep=0pt]
\useasboundingbox  (87.5,88.92) rectangle +(8,5.08);
%\path[line width=0mm] (85.00,81.88) rectangle +(13.33,15.12);
\definecolor{L}{rgb}{0,0,0}
\definecolor{F}{rgb}{0,0,0}
\path[line width=0.30mm, draw=L, fill=F] (89.96,89.99) circle (0.54mm);
%\draw(87.00,84.00) node[anchor=base west]{\fontsize{11.23}{13.66}\selectfont $f_1$};
\path[line width=0.30mm, draw=L] (90.00,95.00) -- (90.00,90.50);
\end{tikzpicture}%
\right>+\\
36\left<B^{(2)}_{\sst M}\left(%
\begin{tikzpicture}[scale=0.7,x=1.00mm, y=1.00mm, inner xsep=0pt, inner ysep=0pt, outer xsep=0pt, outer ysep=0pt]%, baseline=0pt]
\useasboundingbox  (83.00,88.92) rectangle +(14.67,5.08);
\definecolor{L}{rgb}{0,0,0}
\definecolor{F}{rgb}{0,0,0}
\path[line width=0.30mm, draw=L, fill=F] (89.87,90.05) circle (0.54mm);
\path[line width=0.30mm, draw=L] (95.00,95.00) -- (90.00,90.00);
\path[line width=0.30mm, draw=L] (85.00,95.00) -- (90.00,90.00);
\draw(87.00,84.00) node[anchor=base west]{\fontsize{11.23}{17.07}\selectfont $\scriptsize{f_1}$};
\end{tikzpicture},%
\begin{tikzpicture}[scale=0.7,x=1.00mm, y=1.00mm, inner xsep=0pt, inner ysep=0pt, outer xsep=0pt, outer ysep=0pt]%, baseline=0pt]
\useasboundingbox  (83.00,88.92) rectangle +(14.67,5.08);
\definecolor{L}{rgb}{0,0,0}
\definecolor{F}{rgb}{0,0,0}
\path[line width=0.30mm, draw=L, fill=F] (89.87,90.05) circle (0.54mm);
\path[line width=0.30mm, draw=L] (95.00,95.00) -- (90.00,90.00);
\path[line width=0.30mm, draw=L] (85.00,95.00) -- (90.00,90.00);
\draw(87.00,84.00) node[anchor=base west]{\fontsize{11.23}{17.07}\selectfont $\scriptsize{f_1}$};
\end{tikzpicture}
\right)(0),\begin{tikzpicture}[scale=0.7,x=1.00mm, y=1.00mm, inner xsep=0pt, inner ysep=0pt, outer xsep=0pt, outer ysep=0pt]
\useasboundingbox  (84.00,88.92) rectangle +(10,5.08);
%\path[line width=0mm] (85.00,81.88) rectangle +(13.33,15.12);
\definecolor{L}{rgb}{0,0,0}
\definecolor{F}{rgb}{0,0,0}
\path[line width=0.30mm, draw=L, fill=F] (89.87,90.05) circle (0.54mm);
\path[line width=0.30mm, draw=L] (95.00,95.00) -- (90.00,90.00);
\path[line width=0.30mm, draw=L] (85.00,95.00) -- (90.00,90.00);
%\draw(87.00,84.00) node[anchor=base west]{\fontsize{11.23}{13.66}\selectfont $f_1$};
\end{tikzpicture}\otimes%
\begin{tikzpicture}[scale=0.7,x=1.00mm, y=1.00mm, inner xsep=0pt, inner ysep=0pt, outer xsep=0pt, outer ysep=0pt]
\useasboundingbox  (85.05,88.92) rectangle +(12,5.08);
%\path[line width=0mm] (85.00,81.88) rectangle +(13.33,15.12);
\definecolor{L}{rgb}{0,0,0}
\definecolor{F}{rgb}{0,0,0}
\path[line width=0.30mm, draw=L, fill=F] (89.87,90.05) circle (0.54mm);
\path[line width=0.30mm, draw=L] (95.00,95.00) -- (90.00,90.00);
\path[line width=0.30mm, draw=L] (85.00,95.00) -- (90.00,90.00);
%\draw(87.00,84.00) node[anchor=base west]{\fontsize{11.23}{13.66}\selectfont $f_1$};
\end{tikzpicture}%
\right>+\textrm{constant and linear terms}
\,,
\end{multline}
since the co-product acts as:
\[
\Delta\left(\begin{tikzpicture}[scale=0.7,x=1.00mm, y=1.00mm, inner xsep=0pt, inner ysep=0pt, outer xsep=0pt, outer ysep=0pt]%, baseline=0pt]
\useasboundingbox  (83.00,88.92) rectangle +(14.67,5.08);
\definecolor{L}{rgb}{0,0,0}
\definecolor{F}{rgb}{0,0,0}
%\path[line width=1mm, draw=L,anchor=base] (83.00,80.92) rectangle +(10.67,10.08);
\path[line width=0.30mm, draw=L, fill=F, anchor=base] (89.87,90.05) circle (0.54mm);
\path[line width=0.30mm, draw=L] (95.00,95.00) -- (85.00,85.00);
\path[line width=0.30mm, draw=L] (85.00,95.00) -- (95.00,85.00);
\end{tikzpicture}\right)=1\otimes \begin{tikzpicture}[scale=0.7,x=1.00mm, y=1.00mm, inner xsep=0pt, inner ysep=0pt, outer xsep=0pt, outer ysep=0pt]%, baseline=0pt]
\useasboundingbox  (83.00,88.92) rectangle +(14.67,5.08);
\definecolor{L}{rgb}{0,0,0}
\definecolor{F}{rgb}{0,0,0}
%\path[line width=1mm, draw=L,anchor=base] (83.00,80.92) rectangle +(10.67,10.08);
\path[line width=0.30mm, draw=L, fill=F, anchor=base] (89.87,90.05) circle (0.54mm);
\path[line width=0.30mm, draw=L] (95.00,95.00) -- (85.00,85.00);
\path[line width=0.30mm, draw=L] (85.00,95.00) -- (95.00,85.00);
\end{tikzpicture}+\begin{tikzpicture}[scale=0.7,x=1.00mm, y=1.00mm, inner xsep=0pt, inner ysep=0pt, outer xsep=0pt, outer ysep=0pt]%, baseline=0pt]
\useasboundingbox  (83.00,88.92) rectangle +(14.67,5.08);
\definecolor{L}{rgb}{0,0,0}
\definecolor{F}{rgb}{0,0,0}
%\path[line width=1mm, draw=L,anchor=base] (83.00,80.92) rectangle +(10.67,10.08);
\path[line width=0.30mm, draw=L, fill=F, anchor=base] (89.87,90.05) circle (0.54mm);
\path[line width=0.30mm, draw=L] (95.00,95.00) -- (85.00,85.00);
\path[line width=0.30mm, draw=L] (85.00,95.00) -- (95.00,85.00);
\end{tikzpicture}\otimes 1+4\begin{tikzpicture}[scale=0.7,x=1.00mm, y=1.00mm, inner xsep=0pt, inner ysep=0pt, outer xsep=0pt, outer ysep=0pt]%, baseline=0pt]
\useasboundingbox  (83.00,88.92) rectangle +(13,5.08);
\definecolor{L}{rgb}{0,0,0}
\definecolor{F}{rgb}{0,0,0}
\path[line width=0.30mm, draw=L, fill=F] (89.87,90.05) circle (0.54mm);
\path[line width=0.30mm, draw=L] (95.00,95.00) -- (90.00,90.00);
\path[line width=0.30mm, draw=L] (85.00,95.00) -- (90.00,90.00);
\path[line width=0.30mm, draw=L] (90.00,95.00) -- (90.00,90.50);
\end{tikzpicture}\otimes%
\begin{tikzpicture}[scale=0.7,x=1.00mm, y=1.00mm, inner xsep=0pt, inner ysep=0pt, outer xsep=0pt, outer ysep=0pt]%, baseline=0pt]
\useasboundingbox  (86.00,88.92) rectangle +(9,5.08);
\definecolor{L}{rgb}{0,0,0}
\definecolor{F}{rgb}{0,0,0}
\path[line width=0.30mm, draw=L, fill=F] (89.87,90.05) circle (0.54mm);
\path[line width=0.30mm, draw=L] (90.00,95.00) -- (90.00,90.50);
\end{tikzpicture}+%
4\begin{tikzpicture}[scale=0.7,x=1.00mm, y=1.00mm, inner xsep=0pt, inner ysep=0pt, outer xsep=0pt, outer ysep=0pt]%, baseline=0pt]
\useasboundingbox  (84.00,88.92) rectangle +(11,5.08);
\definecolor{L}{rgb}{0,0,0}
\definecolor{F}{rgb}{0,0,0}
\path[line width=0.30mm, draw=L, fill=F] (89.87,90.05) circle (0.54mm);
\path[line width=0.30mm, draw=L] (90.00,95.00) -- (90.00,90.50);
\end{tikzpicture}\otimes%
\begin{tikzpicture}[scale=0.7,x=1.00mm, y=1.00mm, inner xsep=0pt, inner ysep=0pt, outer xsep=0pt, outer ysep=0pt]%, baseline=0pt]
\useasboundingbox  (83.00,88.92) rectangle +(13,5.08);
\definecolor{L}{rgb}{0,0,0}
\definecolor{F}{rgb}{0,0,0}
\path[line width=0.30mm, draw=L, fill=F] (89.87,90.05) circle (0.54mm);
\path[line width=0.30mm, draw=L] (95.00,95.00) -- (90.00,90.00);
\path[line width=0.30mm, draw=L] (85.00,95.00) -- (90.00,90.00);
\path[line width=0.30mm, draw=L] (90.00,95.00) -- (90.00,90.50);
\end{tikzpicture}%
+6\begin{tikzpicture}[scale=0.7,x=1.00mm, y=1.00mm, inner xsep=0pt, inner ysep=0pt, outer xsep=0pt, outer ysep=0pt]%, baseline=0pt]
\useasboundingbox  (83.00,88.92) rectangle +(13,5.08);
\definecolor{L}{rgb}{0,0,0}
\definecolor{F}{rgb}{0,0,0}
\path[line width=0.30mm, draw=L, fill=F] (89.87,90.05) circle (0.54mm);
\path[line width=0.30mm, draw=L] (95.00,95.00) -- (90.00,90.00);
\path[line width=0.30mm, draw=L] (85.00,95.00) -- (90.00,90.00);
\end{tikzpicture}\otimes%
\begin{tikzpicture}[scale=0.7,x=1.00mm, y=1.00mm, inner xsep=0pt, inner ysep=0pt, outer xsep=0pt, outer ysep=0pt]%, baseline=0pt]
\useasboundingbox  (84.00,88.92) rectangle +(9,5.08);
\definecolor{L}{rgb}{0,0,0}
\definecolor{F}{rgb}{0,0,0}
\path[line width=0.30mm, draw=L, fill=F] (89.87,90.05) circle (0.54mm);
\path[line width=0.30mm, draw=L] (95.00,95.00) -- (90.00,90.00);
\path[line width=0.30mm, draw=L] (85.00,95.00) -- (90.00,90.00);
\end{tikzpicture}
\]
It follows from \eqref{B2M} now the graphs contributing to $B_{\sst M}^{(2)}$ are
\[
\Gamma_1=\begin{tikzpicture}[scale=0.7,x=1.00mm, y=1.00mm, inner xsep=0pt, inner ysep=0pt, outer xsep=0pt, outer ysep=0pt]
\useasboundingbox (84.00,88) rectangle +(20,5.08);
\path[line width=0mm] (82.33,85.75) rectangle +(20.09,8.50);
\definecolor{L}{rgb}{0,0,0}
\definecolor{F}{rgb}{0,0,0}
\path[line width=0.30mm, draw=L, fill=F] (84.87,90.05) circle (0.54mm);
\path[line width=0.30mm, draw=L, fill=F] (99.87,90.05) circle (0.54mm);
\path[line width=0.30mm, draw=L] (85.00,90.00) .. controls (89.50,87.00) and (95.50,87.00) .. (100.00,90.00);
\path[line width=0.30mm, draw=L] (85.00,90.00) .. controls (89.50,93.00) and (95.50,93.00) .. (100.00,90.00);
\path[line width=0.30mm, draw=L] (85.00,90.00);
\path[line width=0.30mm, draw=L] (85.00,90.00) -- (100.00,90.00);
\end{tikzpicture}\!\!,\quad
\Gamma_2=\begin{tikzpicture}[scale=0.7,x=1.00mm, y=1.00mm, inner xsep=0pt, inner ysep=0pt, outer xsep=0pt, outer ysep=0pt]
\useasboundingbox  (84.00,88) rectangle +(9,5.08);
\path[line width=0mm] (82.33,85.75) rectangle +(20.09,8.50);
\definecolor{L}{rgb}{0,0,0}
\definecolor{F}{rgb}{0,0,0}
\path[line width=0.30mm, draw=L, fill=F] (84.87,90.05) circle (0.54mm);
\path[line width=0.30mm, draw=L, fill=F] (99.87,90.05) circle (0.54mm);
\path[line width=0.30mm, draw=L] (85.00,90.00) .. controls (89.50,87.00) and (95.50,87.00) .. (100.00,90.00);
\path[line width=0.30mm, draw=L] (85.00,90.00) .. controls (89.50,93.00) and (95.50,93.00) .. (100.00,90.00);
\end{tikzpicture}
\] 
Hence, neglecting constant and linear terms:
\[
B^{(2)}_{\sst M}\left(\begin{tikzpicture}[scale=0.7,x=1.00mm, y=1.00mm, inner xsep=0pt, inner ysep=0pt, outer xsep=0pt, outer ysep=0pt]%, baseline=0pt]
\useasboundingbox  (83.00,88.92) rectangle +(14.67,5.08);
\definecolor{L}{rgb}{0,0,0}
\definecolor{F}{rgb}{0,0,0}
%\path[line width=1mm, draw=L,anchor=base] (83.00,80.92) rectangle +(10.67,10.08);
\path[line width=0.30mm, draw=L, fill=F, anchor=base] (89.87,90.05) circle (0.54mm);
\path[line width=0.30mm, draw=L] (95.00,95.00) -- (85.00,85.00);
\path[line width=0.30mm, draw=L] (85.00,95.00) -- (95.00,85.00);
\draw(87.00,84.00) node[anchor=base west]{\fontsize{11.23}{17.07}\selectfont $\footnotesize{f_1}$};
\end{tikzpicture},%
\begin{tikzpicture}[scale=0.7,x=1.00mm, y=1.00mm, inner xsep=0pt, inner ysep=0pt, outer xsep=0pt, outer ysep=0pt]%, baseline=0pt]
\useasboundingbox  (83.00,88.92) rectangle +(14.67,5.08);
\definecolor{L}{rgb}{0,0,0}
\definecolor{F}{rgb}{0,0,0}
%\path[line width=1mm, draw=L,anchor=base] (83.00,80.92) rectangle +(10.67,10.08);
\path[line width=0.30mm, draw=L, fill=F, anchor=base] (89.87,90.05) circle (0.54mm);
\path[line width=0.30mm, draw=L] (95.00,95.00) -- (85.00,85.00);
\path[line width=0.30mm, draw=L] (85.00,95.00) -- (95.00,85.00);
\draw(87.00,84.00) node[anchor=base west]{\fontsize{11.23}{17.07}\selectfont $\footnotesize{f_1}$};
\end{tikzpicture}
\right)=\left<(f_1\otimes f_1)\cdot b^{\Gamma_1},\begin{tikzpicture}[scale=0.7,x=1.00mm, y=1.00mm, inner xsep=0pt, inner ysep=0pt, outer xsep=0pt, outer ysep=0pt]
\useasboundingbox  (86.00,88.92) rectangle +(7,5.08);
%\path[line width=0mm] (85.00,81.88) rectangle +(13.33,15.12);
\definecolor{L}{rgb}{0,0,0}
\definecolor{F}{rgb}{0,0,0}
\path[line width=0.30mm, draw=L, fill=F] (89.96,89.99) circle (0.54mm);
%\draw(87.00,84.00) node[anchor=base west]{\fontsize{11.23}{13.66}\selectfont $f_1$};
\path[line width=0.30mm, draw=L] (90.00,95.00) -- (90.00,90.50);
\end{tikzpicture}\otimes%
\begin{tikzpicture}[scale=0.7,x=1.00mm, y=1.00mm, inner xsep=0pt, inner ysep=0pt, outer xsep=0pt, outer ysep=0pt]
\useasboundingbox  (87.5,88.92) rectangle +(8,5.08);
%\path[line width=0mm] (85.00,81.88) rectangle +(13.33,15.12);
\definecolor{L}{rgb}{0,0,0}
\definecolor{F}{rgb}{0,0,0}
\path[line width=0.30mm, draw=L, fill=F] (89.96,89.99) circle (0.54mm);
%\draw(87.00,84.00) node[anchor=base west]{\fontsize{11.23}{13.66}\selectfont $f_1$};
\path[line width=0.30mm, draw=L] (90.00,95.00) -- (90.00,90.50);
\end{tikzpicture}%
\right>+\left<(f_1\otimes f_1)\cdot b^{\Gamma_2},\begin{tikzpicture}[scale=0.7,x=1.00mm, y=1.00mm, inner xsep=0pt, inner ysep=0pt, outer xsep=0pt, outer ysep=0pt]
\useasboundingbox  (84.00,88.92) rectangle +(10,5.08);
%\path[line width=0mm] (85.00,81.88) rectangle +(13.33,15.12);
\definecolor{L}{rgb}{0,0,0}
\definecolor{F}{rgb}{0,0,0}
\path[line width=0.30mm, draw=L, fill=F] (89.87,90.05) circle (0.54mm);
\path[line width=0.30mm, draw=L] (95.00,95.00) -- (90.00,90.00);
\path[line width=0.30mm, draw=L] (85.00,95.00) -- (90.00,90.00);
%\draw(87.00,84.00) node[anchor=base west]{\fontsize{11.23}{13.66}\selectfont $f_1$};
\end{tikzpicture}\otimes%
\begin{tikzpicture}[scale=0.7,x=1.00mm, y=1.00mm, inner xsep=0pt, inner ysep=0pt, outer xsep=0pt, outer ysep=0pt]
\useasboundingbox  (85.05,88.92) rectangle +(12,5.08);
%\path[line width=0mm] (85.00,81.88) rectangle +(13.33,15.12);
\definecolor{L}{rgb}{0,0,0}
\definecolor{F}{rgb}{0,0,0}
\path[line width=0.30mm, draw=L, fill=F] (89.87,90.05) circle (0.54mm);
\path[line width=0.30mm, draw=L] (95.00,95.00) -- (90.00,90.00);
\path[line width=0.30mm, draw=L] (85.00,95.00) -- (90.00,90.00);
%\draw(87.00,84.00) node[anchor=base west]{\fontsize{11.23}{13.66}\selectfont $f_1$};
\end{tikzpicture}\right>\,.
\]
A similar reasoning leads to
\[
B^{(2)}_{\sst M}\left(\begin{tikzpicture}[scale=0.7,x=1.00mm, y=1.00mm, inner xsep=0pt, inner ysep=0pt, outer xsep=0pt, outer ysep=0pt]%, baseline=0pt]
\useasboundingbox  (83.00,88.92) rectangle +(14.67,5.08);
\definecolor{L}{rgb}{0,0,0}
\definecolor{F}{rgb}{0,0,0}
%\path[line width=1mm, draw=L,anchor=base] (83.00,80.92) rectangle +(10.67,10.08);
\path[line width=0.30mm, draw=L, fill=F, anchor=base] (89.87,90.05) circle (0.54mm);
\path[line width=0.30mm, draw=L] (95.00,95.00) -- (85.00,85.00);
\path[line width=0.30mm, draw=L] (85.00,95.00) -- (95.00,85.00);
\draw(87.00,84.00) node[anchor=base west]{\fontsize{11.23}{17.07}\selectfont $\footnotesize{f_1}$};
\end{tikzpicture},%
\begin{tikzpicture}[scale=0.7,x=1.00mm, y=1.00mm, inner xsep=0pt, inner ysep=0pt, outer xsep=0pt, outer ysep=0pt]%, baseline=0pt]
\useasboundingbox  (83.00,88.92) rectangle +(14.67,5.08);
\definecolor{L}{rgb}{0,0,0}
\definecolor{F}{rgb}{0,0,0}
\path[line width=0.30mm, draw=L, fill=F] (89.87,90.05) circle (0.54mm);
\path[line width=0.30mm, draw=L] (95.00,95.00) -- (90.00,90.00);
\path[line width=0.30mm, draw=L] (85.00,95.00) -- (90.00,90.00);
\draw(87.00,84.00) node[anchor=base west]{\fontsize{11.23}{17.07}\selectfont $\scriptsize{f_2}$};
\end{tikzpicture}
\right)=6\left<B^{(2)}_{\sst M}\left(%
\begin{tikzpicture}[scale=0.7,x=1.00mm, y=1.00mm, inner xsep=0pt, inner ysep=0pt, outer xsep=0pt, outer ysep=0pt]%, baseline=0pt]
\useasboundingbox  (83.00,88.92) rectangle +(14.67,5.08);
\definecolor{L}{rgb}{0,0,0}
\definecolor{F}{rgb}{0,0,0}
\path[line width=0.30mm, draw=L, fill=F] (89.87,90.05) circle (0.54mm);
\path[line width=0.30mm, draw=L] (95.00,95.00) -- (90.00,90.00);
\path[line width=0.30mm, draw=L] (85.00,95.00) -- (90.00,90.00);
\draw(87.00,84.00) node[anchor=base west]{\fontsize{11.23}{17.07}\selectfont $\scriptsize{f_1}$};
\end{tikzpicture},%
\begin{tikzpicture}[scale=0.7,x=1.00mm, y=1.00mm, inner xsep=0pt, inner ysep=0pt, outer xsep=0pt, outer ysep=0pt]%, baseline=0pt]
\useasboundingbox  (83.00,88.92) rectangle +(14.67,5.08);
\definecolor{L}{rgb}{0,0,0}
\definecolor{F}{rgb}{0,0,0}
\path[line width=0.30mm, draw=L, fill=F] (89.87,90.05) circle (0.54mm);
\path[line width=0.30mm, draw=L] (95.00,95.00) -- (90.00,90.00);
\path[line width=0.30mm, draw=L] (85.00,95.00) -- (90.00,90.00);
\draw(87.00,84.00) node[anchor=base west]{\fontsize{11.23}{17.07}\selectfont $\scriptsize{f_2}$};
\end{tikzpicture}
\right)(0),\begin{tikzpicture}[scale=0.7,x=1.00mm, y=1.00mm, inner xsep=0pt, inner ysep=0pt, outer xsep=0pt, outer ysep=0pt]
\useasboundingbox  (84.00,88.92) rectangle +(10,5.08);
%\path[line width=0mm] (85.00,81.88) rectangle +(13.33,15.12);
\definecolor{L}{rgb}{0,0,0}
\definecolor{F}{rgb}{0,0,0}
\path[line width=0.30mm, draw=L, fill=F] (89.87,90.05) circle (0.54mm);
\path[line width=0.30mm, draw=L] (95.00,95.00) -- (90.00,90.00);
\path[line width=0.30mm, draw=L] (85.00,95.00) -- (90.00,90.00);
%\draw(87.00,84.00) node[anchor=base west]{\fontsize{11.23}{13.66}\selectfont $f_1$};
\end{tikzpicture}\otimes1%
\right>=6\left<(f_1\otimes f_2)\cdot b^{\Gamma_2},\begin{tikzpicture}[scale=0.7,x=1.00mm, y=1.00mm, inner xsep=0pt, inner ysep=0pt, outer xsep=0pt, outer ysep=0pt]
\useasboundingbox  (84.00,88.92) rectangle +(10,5.08);
%\path[line width=0mm] (85.00,81.88) rectangle +(13.33,15.12);
\definecolor{L}{rgb}{0,0,0}
\definecolor{F}{rgb}{0,0,0}
\path[line width=0.30mm, draw=L, fill=F] (89.87,90.05) circle (0.54mm);
\path[line width=0.30mm, draw=L] (95.00,95.00) -- (90.00,90.00);
\path[line width=0.30mm, draw=L] (85.00,95.00) -- (90.00,90.00);
%\draw(87.00,84.00) node[anchor=base west]{\fontsize{11.23}{13.66}\selectfont $f_1$};
\end{tikzpicture}\otimes1%
\right>
\]
and
\[
B^{(2)}_{\sst M}\left(\begin{tikzpicture}[scale=0.7,x=1.00mm, y=1.00mm, inner xsep=0pt, inner ysep=0pt, outer xsep=0pt, outer ysep=0pt]%, baseline=0pt]
\useasboundingbox  (83.00,88.92) rectangle +(14.67,5.08);
\definecolor{L}{rgb}{0,0,0}
\definecolor{F}{rgb}{0,0,0}
%\path[line width=1mm, draw=L,anchor=base] (83.00,80.92) rectangle +(10.67,10.08);
\path[line width=0.30mm, draw=L, fill=F, anchor=base] (89.87,90.05) circle (0.54mm);
\path[line width=0.30mm, draw=L] (95.00,95.00) -- (85.00,85.00);
\path[line width=0.30mm, draw=L] (85.00,95.00) -- (95.00,85.00);
\draw(87.00,84.00) node[anchor=base west]{\fontsize{11.23}{17.07}\selectfont $\footnotesize{f_1}$};
\end{tikzpicture},%
\begin{tikzpicture}[scale=0.7,x=1.00mm, y=1.00mm, inner xsep=0pt, inner ysep=0pt, outer xsep=0pt, outer ysep=0pt]%, baseline=0pt]
\useasboundingbox  (83.00,88.92) rectangle +(14.67,5.08);
\definecolor{L}{rgb}{0,0,0}
\definecolor{F}{rgb}{0,0,0}
\path[line width=0.30mm, draw=L, fill=F] (89.87,90.05) circle (0.54mm);
\path[line width=0.30mm, draw=L] (95.00,95.00) -- (90.00,90.00);
\path[line width=0.30mm, draw=L] (85.00,95.00) -- (90.00,90.00);
\draw(88.00,84.00) node[anchor=base west]{\fontsize{11.23}{17.07}\selectfont $f_3$};
\path[line width=0.30mm, draw=L] (87.00,91.00);
\path[line width=0.30mm, draw=L] (87.00,91.00) -- (89.00,93.00);
\path[line width=0.30mm, draw=L] (91.00,93.00) -- (93.00,91.00);
\end{tikzpicture}
\right)=6\left<B^{(2)}_{\sst M}\left(%
\begin{tikzpicture}[scale=0.7,x=1.00mm, y=1.00mm, inner xsep=0pt, inner ysep=0pt, outer xsep=0pt, outer ysep=0pt]%, baseline=0pt]
\useasboundingbox  (83.00,88.92) rectangle +(14.67,5.08);
\definecolor{L}{rgb}{0,0,0}
\definecolor{F}{rgb}{0,0,0}
\path[line width=0.30mm, draw=L, fill=F] (89.87,90.05) circle (0.54mm);
\path[line width=0.30mm, draw=L] (95.00,95.00) -- (90.00,90.00);
\path[line width=0.30mm, draw=L] (85.00,95.00) -- (90.00,90.00);
\draw(87.00,84.00) node[anchor=base west]{\fontsize{11.23}{17.07}\selectfont $\scriptsize{f_1}$};
\end{tikzpicture},%
\begin{tikzpicture}[scale=0.7,x=1.00mm, y=1.00mm, inner xsep=0pt, inner ysep=0pt, outer xsep=0pt, outer ysep=0pt]%, baseline=0pt]
\useasboundingbox  (83.00,88.92) rectangle +(14.67,5.08);
\definecolor{L}{rgb}{0,0,0}
\definecolor{F}{rgb}{0,0,0}
\path[line width=0.30mm, draw=L, fill=F] (89.87,90.05) circle (0.54mm);
\path[line width=0.30mm, draw=L] (95.00,95.00) -- (90.00,90.00);
\path[line width=0.30mm, draw=L] (85.00,95.00) -- (90.00,90.00);
\draw(88.00,84.00) node[anchor=base west]{\fontsize{11.23}{17.07}\selectfont $f_3$};
\path[line width=0.30mm, draw=L] (87.00,91.00);
\path[line width=0.30mm, draw=L] (87.00,91.00) -- (89.00,93.00);
\path[line width=0.30mm, draw=L] (91.00,93.00) -- (93.00,91.00);
\end{tikzpicture}
\right)(0),\begin{tikzpicture}[scale=0.7,x=1.00mm, y=1.00mm, inner xsep=0pt, inner ysep=0pt, outer xsep=0pt, outer ysep=0pt]
\useasboundingbox  (84.00,88.92) rectangle +(10,5.08);
%\path[line width=0mm] (85.00,81.88) rectangle +(13.33,15.12);
\definecolor{L}{rgb}{0,0,0}
\definecolor{F}{rgb}{0,0,0}
\path[line width=0.30mm, draw=L, fill=F] (89.87,90.05) circle (0.54mm);
\path[line width=0.30mm, draw=L] (95.00,95.00) -- (90.00,90.00);
\path[line width=0.30mm, draw=L] (85.00,95.00) -- (90.00,90.00);
%\draw(87.00,84.00) node[anchor=base west]{\fontsize{11.23}{13.66}\selectfont $f_1$};
\end{tikzpicture}\otimes1%
\right>=6\left<(f_1\otimes f_3)\cdot b^{\Gamma_3},\begin{tikzpicture}[scale=0.7,x=1.00mm, y=1.00mm, inner xsep=0pt, inner ysep=0pt, outer xsep=0pt, outer ysep=0pt]
\useasboundingbox  (84.00,88.92) rectangle +(10,5.08);
%\path[line width=0mm] (85.00,81.88) rectangle +(13.33,15.12);
\definecolor{L}{rgb}{0,0,0}
\definecolor{F}{rgb}{0,0,0}
\path[line width=0.30mm, draw=L, fill=F] (89.87,90.05) circle (0.54mm);
\path[line width=0.30mm, draw=L] (95.00,95.00) -- (90.00,90.00);
\path[line width=0.30mm, draw=L] (85.00,95.00) -- (90.00,90.00);
%\draw(87.00,84.00) node[anchor=base west]{\fontsize{11.23}{13.66}\selectfont $f_1$};
\end{tikzpicture}\otimes1%
\right>
\]
In the latter case there is a new graph appearing, namely
\[
\Gamma_3=\begin{tikzpicture}[scale=0.7,x=1.00mm, y=1.00mm, inner xsep=0pt, inner ysep=0pt, outer xsep=0pt, outer ysep=0pt]
\useasboundingbox (84.00,88) rectangle +(20,5.08);
\path[line width=0mm] (82.33,85.75) rectangle +(20.09,8.50);
\definecolor{L}{rgb}{0,0,0}
\definecolor{F}{rgb}{0,0,0}
\path[line width=0.30mm, draw=L, fill=F] (84.87,90.05) circle (0.54mm);
\path[line width=0.30mm, draw=L, fill=F] (99.87,90.05) circle (0.54mm);
\path[line width=0.30mm, draw=L] (85.00,90.00) .. controls (89.50,87.00) and (95.50,87.00) .. (100.00,90.00);
\path[line width=0.30mm, draw=L] (85.00,90.00) .. controls (89.50,93.00) and (95.50,93.00) .. (100.00,90.00);
\path[line width=0.30mm, draw=L] (93.00,92.90);
\path[line width=0.30mm, draw=L] (93.00,93.10) -- (91.60,91.30);
\path[line width=0.30mm, draw=L] (93.10,88.60) -- (91.70,86.80);
\end{tikzpicture}\,.
\]
Calculating $B_{\sst M}^{(2)}$ is now reduced to finding the residues: $\Res u^{\Gamma i}$, $i=1,2,3$. The (rather lengthy) computation can be found in section 7.2 of \cite{BDF}.

From the point of view of Kontsevich-Zagier periods, one gets some more interesting numbers in calculating higher orders of $B$. In particular, the wheel with three spokes appears as a contribution to
\[
B^{(4)}_{\sst M}\left( \begin{tikzpicture}[scale=0.7,x=1.00mm, y=1.00mm, inner xsep=0pt, inner ysep=0pt, outer xsep=0pt, outer ysep=0pt]%, baseline=0pt]
\useasboundingbox  (83.00,88.92) rectangle +(14.67,5.08);
\definecolor{L}{rgb}{0,0,0}
\definecolor{F}{rgb}{0,0,0}
%\path[line width=1mm, draw=L,anchor=base] (83.00,80.92) rectangle +(10.67,10.08);
\path[line width=0.30mm, draw=L, fill=F, anchor=base] (89.87,90.05) circle (0.54mm);
\path[line width=0.30mm, draw=L] (95.00,95.00) -- (85.00,85.00);
\path[line width=0.30mm, draw=L] (85.00,95.00) -- (95.00,85.00);
\draw(87.00,84.00) node[anchor=base west]{\fontsize{11.23}{17.07}\selectfont $\footnotesize{f_1}$};
\end{tikzpicture}^{\otimes 4}\right)=2^8\left<f_1^{\otimes 4}\,b^{\Gamma_4},\begin{tikzpicture}[scale=0.7,x=1.00mm, y=1.00mm, inner xsep=0pt, inner ysep=0pt, outer xsep=0pt, outer ysep=0pt]
\useasboundingbox  (86.00,88.92) rectangle +(7,5.08);
%\path[line width=0mm] (85.00,81.88) rectangle +(13.33,15.12);
\definecolor{L}{rgb}{0,0,0}
\definecolor{F}{rgb}{0,0,0}
\path[line width=0.30mm, draw=L, fill=F] (89.96,89.99) circle (0.54mm);
%\draw(87.00,84.00) node[anchor=base west]{\fontsize{11.23}{13.66}\selectfont $f_1$};
\path[line width=0.30mm, draw=L] (90.00,95.00) -- (90.00,90.50);
\end{tikzpicture}^{\otimes 4}\right>+\dots\,,
\]
where
\vspace{-2ex} 
\[
\Gamma_4=\begin{tikzpicture}
[scale=0.5,x=1.00mm, y=1.00mm, inner xsep=0pt, inner ysep=0pt, outer xsep=0pt, outer ysep=0pt]
\useasboundingbox (65.92,90) rectangle +(28.17,28.68);
\definecolor{L}{rgb}{0,0,0}
\definecolor{F}{rgb}{0,0,0}
\path[line width=0.30mm, draw=L, fill=F] (79.87,104.05) circle (0.54mm);
\path[line width=0.30mm, draw=L] (80.00,92.00) circle (12.08mm);
\path[line width=0.30mm, draw=L] (80.00,92.00) circle (12.08mm);
\path[line width=0.30mm, draw=L, fill=F] (89.87,85.05) circle (0.54mm);
\path[line width=0.30mm, draw=L, fill=F] (69.87,85.05) circle (0.54mm);
\path[line width=0.30mm, draw=L, fill=F] (79.87,92.05) circle (0.54mm);
\path[line width=0.30mm, draw=L] (80.00,104.00);
\path[line width=0.30mm, draw=L] (80.00,104.00) -- (80.00,92.00);
\path[line width=0.30mm, draw=L] (70.00,85.00) -- (80.00,92.00);
\path[line width=0.30mm, draw=L] (80.00,92.00) -- (90.00,85.00);
\end{tikzpicture}
\]
and $b^{\Gamma_4}=\Res u^{\Gamma_4}$.
\end{example}

\section{Conclusion}
In this paper we reviewed some important algebraic structures appearing in perturbative Algebraic Quantum Field Theory (pAQFT) on Minkowski spacetime \cite{BDF} and we have shown how these relate to periods, usually investigated in a different context in Euclidean QFT in momentum space. The approach we advocate here provides a natural interpretation of these periods both in the mathematical and physical context. Mathematically, these correspond to distributional residues and are therefore intrinsic characterizations of scaling properties of certain class of distributions. Physically, they are relevant in computing the $\beta$-function. Note that, in our approach, the later characterization is independent of any regularization scheme. In fact, regularization is not needed at all and there is no need to recur to ill defined divergent expressions. Instead, the whole analysis is centered around the singularity structure of distributions that arise from taking powers of the Feynman propagator.

The main result of this paper is that distributional residues in pAQFT, corresponding to CK primitive graphs, are up to a factor that we compute, the same as Feynman periods in the CK framework (as conjectured in \cite{BDF}). The remaining EG primitive graphs, which are not CK primitive, also give rise to multiples of the same periods.

For the future research it would be worth investigating the distributional residues arising in pAQFT on other Lorentzian manifolds. Some interesting results have already been obtained for de Sitter spacetime in \cite{HdeSitter}. All the fundamental structures of pAQFT presented in this paper generalize easily to curved spacetimes. The only difference is the form of the Feynman propagator (or rather the ``Feynman-like'' propagator $H^{\mathrm{F}}$). The hope is that looking at more general propagators one would obtain a richer structure of residues and some new periods appearing, which are not present in the Minkowski spacetime context. It would be also interesting to investigate how these relate to motives.

%In the present paper I will show how this point of view can lead to identifying some intrinsic connections between number theory and QFT, which are independent of the choice of the regularization scheme.
%\newpage
\section*{Acknowledgments}
I would like to thank ICMAT (Madrid) for hospitality and financial support. The ideas presented in this paper were developed during my stay in Madrid in 2014 as part of the ``Research Trimester on Multiple Zeta Values, Multiple Polylogarithms, and Quantum Field Theory''.

\bibliographystyle{amsalpha}
\bibliography{References}

\end{document}